\documentclass[10pt,journal]{IEEEtran}

\usepackage{amsmath,amssymb,amsthm}
\usepackage{mathtools}
\usepackage{booktabs}
\usepackage{array}
\usepackage{tikz}
\usetikzlibrary{positioning,arrows.meta,calc,fit,backgrounds,shapes.geometric}
\usepackage{algorithm}
\usepackage{algpseudocode}
\usepackage{enumitem}
\usepackage{float}
\usepackage{graphicx}
\usepackage[colorlinks=true,allcolors=blue]{hyperref}

\theoremstyle{plain}
\newtheorem{theorem}{Theorem}
\newtheorem{lemma}{Lemma}
\newtheorem{proposition}{Proposition}
\newtheorem{corollary}{Corollary}
\theoremstyle{definition}
\newtheorem{definition}{Definition}
\newtheorem{example}{Example}
\newtheorem{remark}{Remark}

\newcommand{\F}{\mathbb{F}}
\newcommand{\R}{R}
\newcommand{\wt}{\operatorname{wt}}
\newcommand{\supp}{\operatorname{supp}}
\newcommand{\ann}{\operatorname{ann}}
\newcommand{\HX}{H_X}
\newcommand{\HZ}{H_Z}
\newcommand{\code}[3]{[\![#1,#2,#3]\!]}
\newcommand{\cln}{\,{:}\,}

\begin{document}

\title{Logical Operator Decomposition for Distance\\ Analysis of Bivariate Bicycle Codes}

\author{Mohammad Rowshan and Simon Devitt%
\thanks{M. Rowshan and S. Devitt are with the Centre for Quantum Software and Information, School of Computer Science, University of Technology Sydney, NSW 2007, Australia (e-mail: mohammad.rowshan@uts.edu.au, simon.devitt@uts.edu.au).}}

\markboth{Rowshan and Devitt: Logical Operator Decomposition for Distance Analysis of Bivariate Bicycle Codes}%
{Rowshan and Devitt: Logical Operator Decomposition for Distance Analysis of Bivariate Bicycle Codes}

\maketitle

\begin{abstract}
Bivariate bicycle (BB) quantum codes are a prominent finite-length family of quantum low-density parity-check codes, but their minimum distance is usually established numerically rather than read from the defining polynomials. We study the $Z$-logical quotient $K/S$ over $\mathbb F_2[x,y]/(x^\ell-1,y^m-1)$ and show that it fits into a short exact sequence with an annihilator quotient as kernel and a colon quotient as cokernel. The sequence gives an explicit logical basis, a dimension formula, and a componentwise distance identity $d_Z=\min(d_{\mathcal A},d_{\mathcal C})$. Using the Frobenius structure of the finite group algebra, we prove $r_{\mathcal A}=r_{\mathcal C}=k/2$ for every BB code, including repeated-root cases. The algebraic component of a logical class is distinct from the support shape of its lightest representatives: an annihilator class can have a lighter two-block representative, and a colon class can have a one-sided minimum. For lower bounds we show that every proper subset of a minimum-weight logical operator has nonzero syndrome, and that this property persists inside the colon component but not inside the annihilator component. A translation-anchored cluster search built on it proves the distances $4,6,10,10,12,18$ of the six standard BB codes of lengths $18$ to $288$ and enumerates every minimum-weight logical operator. The resulting census shows that $\code{108}{8}{10}$ is the only one of the six whose distance is attained in a single component, with $d_{\mathcal C}=10$ and $d_{\mathcal A}=12$.
\end{abstract}

\begin{IEEEkeywords}
Quantum LDPC codes, bivariate bicycle codes, minimum distance, syzygy module, annihilator, colon ideal, abelian codes.
\end{IEEEkeywords}

\IEEEpeerreviewmaketitle

\section{Introduction}\label{sec:intro}
\IEEEPARstart{Q}{uantum} error correction protects logical information by spreading it across many physical qubits, and the number of errors a code tolerates is controlled by its minimum distance. The surface code \cite{kitaev2003,fowler2012} reads its distance directly from a planar lattice but encodes at a vanishing rate. Quantum low-density parity-check (LDPC) codes keep the parity checks sparse while raising the rate, and a sequence of constructions has shaped the area: hypergraph product codes \cite{tillich2014}, lifted and balanced product codes with strong finite-length behaviour \cite{panteleev2021,breuckmann2021bp}, asymptotically good quantum LDPC codes \cite{panteleev2022}, and quantum Tanner codes \cite{leverrier2022}. A broad survey is given in \cite{breuckmann2021}.

Bivariate bicycle (BB) codes are prominent finite-length representatives of this line of work. They descend from the bicycle construction of MacKay, Mitchison, and McFadden \cite{mackay2004} and were brought to prominence by Bravyi \emph{et al.} \cite{bravyi2024}, who demonstrated weight-six checks, a layout friendly to a planar architecture, and circuit-level performance competitive with surface-code memory at substantially higher rate. They are a special case of two-block group algebra codes \cite{linpryadko2024}, and their dimension and existence have a clean algebraic description through the defining-polynomial structure \cite{postema2026}. The Gross code $\code{144}{12}{12}$ has become a reference finite-length example because it encodes twelve logical qubits with weight-six checks at substantially higher rate than a comparable surface-code memory \cite{bravyi2024}.

For a BB code, the dimension and the check weight follow from the two defining polynomials by inspection, but the distance does not. In practice, it is established after the fact, by numerical or search-based methods: mixed-integer programming as used by Bravyi \emph{et al.} \cite{bravyi2024} (following integer-program methods such as \cite{landahl2011}), the randomized window method of QDistRnd \cite{qdistrnd2022}, or decoder-assisted sampling \cite{roffe2020}. None of these methods reads the distance from the algebraic structure of $(a,b)$ directly; each treats the code as a fixed object to be probed once it already exists. Finite-length code design therefore tends to build first and measure the distance afterwards, rather than the reverse.

Although the terminology below is drawn from commutative algebra, every object in this paper is a finite-dimensional binary vector space and every computation reduces to Gaussian elimination over $\F_2$.

\subsection{Related work}
The algebraic structure of BB and two-block group algebra codes is well understood at the level of dimension and existence. Lin and Pryadko \cite{linpryadko2024} and Postema and Kokkelmans \cite{postema2026} give the dimension through a greatest common divisor of the defining polynomials, equivalently through the common-zero count of Proposition~\ref{prop:spec}. Eberhardt and Steffan construct explicit logical bases and fold-transversal Clifford gates for BB codes \cite{eberhardt2024}, and structure-driven subfamilies include the coprime codes of Wang and Mueller \cite{wang2026} and the self-dual codes of Liang and Chen \cite{liang2025}. These results describe the code or a convenient logical basis; they do not separate logical classes by how their low-weight representatives arise.

Distance is usually settled numerically. Bravyi \emph{et al.}\ obtained the distances of the standard codes by mixed-integer programming, reporting the length-$288$ value as an upper bound \cite{bravyi2024}; QDistRnd \cite{qdistrnd2022} and decoder-assisted sampling \cite{roffe2020} give upper-bound witnesses. For sparse codes, Dumer, Kovalev, and Pryadko showed that minimum-weight codewords can be grown as connected clusters in the Tanner graph, which makes distance verification exponential in $d$ but nearly independent of $n$ \cite{dumer2017}. Classical defining-set and apparent-distance bounds for abelian codes \cite{bernal2016,macwilliams1977} apply to a single ideal rather than to the pair $(a,b)$. The exact sequence used here is the standard annihilator--colon decomposition of a relation module \cite{eisenbud1995}, specialised to the BB logical quotient; its new use is to label logical classes and to split the distance problem, and Section~\ref{sec:cluster} shows that the splitting interacts nontrivially with cluster verification.

\subsection{Contributions}
The paper has one structural result and two computational consequences. We state them here in the order in which they are used.

\paragraph{Logical quotient.}
Write a $Z$-type operator as a pair $(u,v)\in\R^2$, one polynomial for each qubit block. The logical classes are $K/S$, where $K=\{(u,v):au+bv=0\}$ and $S=\{(br,ar):r\in\R\}$. Theorem~\ref{thm:exact} identifies the kernel and quotient of the map that retains the right block modulo $(a)$:
\begin{equation}
0\longrightarrow
\ann(a)/b\,\ann(a)
\xrightarrow{\ \iota\ }
K/S
\xrightarrow{\ \pi\ }
(a\cln b)/(a)
\longrightarrow0 .
\end{equation}
Here $t\in\ann(a)$ denotes a polynomial satisfying $at=0$, $\iota(t+b\ann(a))=[(t,0)]$, and $\pi[(u,v)]=v+(a)$. Thus the kernel of $\pi$ consists exactly of classes that \emph{admit} a left-only representative. This does not say that their minimum-weight representative is left-only. Likewise, $\pi[(u,v)]\ne0$ labels a colon class but does not force every representative to occupy both blocks. Keeping this algebraic label separate from representative shape is important in the examples below.
Figure~\ref{fig:seq-interp} gives a visual summary of this kernel/quotient view.

\begin{figure}[!t]
\centering
\begin{tikzpicture}[font=\small, node distance=5mm,
  bx/.style={draw, rounded corners, align=center, inner sep=3pt, text width=32mm},
  q/.style={bx, fill=gray!5},
  comp/.style={bx, fill=blue!6, draw=blue!45},
  ar/.style={-{Latex}, gray!70, semithick}]

 \node[bx, text width=42mm] (op) {a logical class\\$[(u,v)]\in K/S$};
 \node[q, below=6mm of op, xshift=-19mm] (q1) {left-only representative?\\ $[(t,0)]$, $t\in\ann(a)$};
 \node[q, below=6mm of op, xshift=19mm]  (q2) {residual right-block\\information $v+(a)$};

 \node[comp, below=of q1] (c1) {annihilator component\\ $\ann(a)/b\,\ann(a)$};
 \node[comp, below=of q2] (c2) {colon component\\ $(a\cln b)/(a)$};

 \draw[ar] (op) -- (q1);
 \draw[ar] (op) -- (q2);
 \draw[ar] (q1) -- (c1);
 \draw[ar] (q2) -- (c2);

 \path (c1) -- (c2) coordinate[midway] (mid);
 \node[font=\footnotesize, black!75, below=4.8mm of mid, align=center, text width=70mm]
   {the exact sequence of Theorem~\ref{thm:exact} guarantees that these two pieces account for all of $K/S$};
\end{tikzpicture}
\caption{Visual reading of the annihilator--colon decomposition. The left branch identifies classes that admit a left-only representative, while the right branch records the residual right-block information measured by $\pi$. The exact sequence relates these two pieces without introducing a third algebraic component.}
\label{fig:seq-interp}
\end{figure}
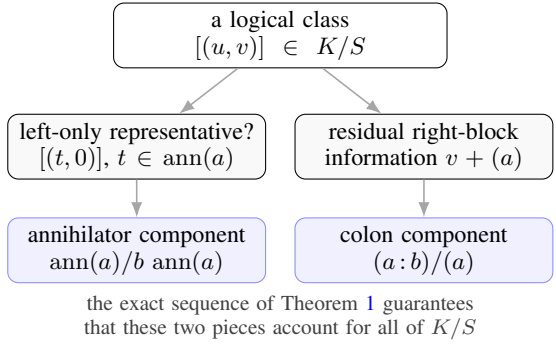

The short exact sequence immediately gives
\begin{align}
 k&=r_{\mathcal A}+r_{\mathcal C},\nonumber\\
 r_{\mathcal A}&=\dim\ann(a)-\dim b\ann(a),\nonumber\\
 r_{\mathcal C}&=\dim(a\cln b)-\dim(a).
\end{align}
We then use the Frobenius pairing of the finite group algebra to prove the stronger identity
\begin{equation}
 r_{\mathcal A}=r_{\mathcal C}=k/2
\end{equation}
for every BB code, without a semisimplicity assumption (Lemma~\ref{lem:frob}). In the semisimple case, Proposition~\ref{prop:spec} recovers the same equality character by character from the common zeros of $a$ and $b$.

\paragraph{Distance and representative shape.}
The exact sequence also partitions the distance calculation by the invariant $\pi$:
\begin{equation}
 d_Z=\min(d_{\mathcal A},d_{\mathcal C})
\end{equation}
(Theorem~\ref{thm:ddec}). The two minima are exact coset-leader problems. They are smaller and structurally labelled, but they are not replaced by Gaussian elimination: minimum weight remains the hard step. We therefore treat \emph{component} (annihilator or colon) and \emph{shape} (one-sided, lopsided, or balanced) as two independent descriptions. The small $\code{18}{4}{4}$ example shows why: $w_{\ann}=6$, yet $d_{\mathcal A}=4$ because an annihilator class has a lighter two-block representative. Conversely, the definitions allow a colon class to have a one-sided representative; Example~\ref{ex:72} gives an explicit instance.

\paragraph{Exact finite-length verification.}
Upper bounds come from explicit logical operators. For lower bounds we prove that every proper subset of a minimum-weight logical operator has nonzero $X$-syndrome (Lemma~\ref{lem:conn}), so minimum-weight logicals can be grown one check at a time from a single anchor qubit fixed by translation symmetry (Theorem~\ref{thm:cluster}). The same property holds for minimum-weight elements of the colon component but fails in general for the annihilator component (Corollary~\ref{cor:colonconn}), which gives $d_{\mathcal C}$ exactly by the same search. On the six standard examples the search proves $d=4,6,10,10,12,18$ for lengths $18$ to $288$, confirming the length-$288$ value reported as a numerical upper bound in \cite{bravyi2024}. Run past the first hit, it enumerates all minimum-weight logicals, and the resulting census (Table~\ref{tab:census}) determines the component and shape of every minimum-weight class. In particular $\code{108}{8}{10}$ has $d_{\mathcal C}=10<d_{\mathcal A}=12$, and all its minimum-weight classes are balanced.

The remainder is organised accordingly. Section~\ref{sec:prelim} gives the ring--matrix dictionary, Section~\ref{sec:decomp} proves the structural results, Section~\ref{sec:tax} separates algebraic component from representative shape, Section~\ref{sec:lb} develops the exact lower bounds, and Section~\ref{sec:results} reports the six finite-length examples, the minimum-weight census, and the design diagnostics.

\section{Preliminaries}\label{sec:prelim}

\begin{table}[!t]
\centering\small
\caption{Notation, its matrix form, and what each object means for the code. $N=\ell m$, $n=2N$.}
\label{tab:notation}
\renewcommand{\arraystretch}{1.25}
\begin{tabular}{@{}>{\raggedright\arraybackslash}p{1.15cm}>{\raggedright\arraybackslash}p{2.0cm}>{\raggedright\arraybackslash}p{4.1cm}@{}}
\toprule
symbol & matrix form & meaning for the code \\
\midrule
$\R$ & $\F_2^N$ & one block of $N$ qubits on the torus \\
$a,b$ & $A,B$ & the two check patterns \\
$\HX$ & $[\,A\mid B\,]$ & $X$-check matrix \\
$K$ & $\ker\HX$ & all $Z$-operators commuting with the $X$-checks \\
$S$ & $\operatorname{im}\HZ^{\mathsf T}$ & $Z$-stabilizers, i.e.\ trivial operators \\
$K/S$ & $\ker\HX/S$ & nontrivial $Z$-logical operators \\
$\ann(a)$ & $\ker A$ & operators living on one block alone \\
$(a\cln b)$ & $B^{-1}(\operatorname{im}A)$ & right-block words the left block can cancel \\
$d_{\mathcal A}$ & --- & minimum weight among classes with $\pi=0$ \\
$d_{\mathcal C}$ & --- & minimum weight among classes with $\pi\ne0$ \\
\bottomrule
\end{tabular}
\end{table}

\noindent Table~\ref{tab:notation} collects the notation used throughout, together with the matrix form of each object and its meaning for the code.

\subsection{Notation and the ambient ring}
We work over the binary field $\F=\F_2$. Fix integers $\ell,m\ge1$ and let
\begin{equation}
\R=\F[x,y]/(x^\ell-1,\;y^m-1)
\end{equation}
be the group algebra of $\mathbb Z_\ell\times\mathbb Z_m$ over $\F$. An element $f\in\R$ is a binary bivariate polynomial with $x$-exponents read modulo $\ell$ and $y$-exponents modulo $m$; the monomials $\{x^iy^j:0\le i<\ell,\,0\le j<m\}$ form a basis, so $\dim_\F\R=\ell m$. Multiplication is cyclic convolution in both indices, and over $\F_2$ squaring is additive, $(f+g)^2=f^2+g^2$. We write $\supp(f)$ for the monomials appearing in $f$ and $\wt(f)=|\supp(f)|$ for its Hamming weight. Multiplication by $f$ is an $\F$-linear map of $\R$; in the monomial basis it is the $\ell m\times\ell m$ binary matrix $F$, a sum of $\wt(f)$ permutation matrices, and these matrices commute because $\R$ is commutative. Every object in this paper is therefore a subspace of a finite $\F_2$-space and is computed by Gaussian elimination.

\subsection{Bivariate bicycle codes and the logical quotient}\label{sec:bbcodes}
Throughout the paper we use $\code{18}{4}{4}$ as a running example: it is small enough that every step, from the parity-check rows to the final distance proof, can be written out and checked by hand. Fix two polynomials $a,b\in\R$, the code's defining checks. The bivariate bicycle code of $(a,b)$ is the Calderbank--Shor--Steane code \cite{calderbank1996,steane1996} with parity-check matrices
\begin{equation}
\HX=[\,A\mid B\,],\qquad \HZ=[\,B^{\mathsf T}\mid A^{\mathsf T}\,],
\end{equation}
on $n=2\ell m$ qubits. A length-$n$ vector splits into a left and a right block, each a copy of $\R$, and a single $X$-check reads $au+bv$ (Fig.~\ref{fig:struct}). Each check has weight $\wt(a)+\wt(b)$, equal to six for the examples considered here. Commutativity of $\R$ gives $\HX\HZ^{\mathsf T}=AB+BA=0$, so the code is well defined.

\begin{figure}[!t]
\centering
\begin{tikzpicture}[scale=0.85,
  q/.style={circle,draw,minimum size=3mm,inner sep=0pt},
  cc/.style={rectangle,draw,fill=black!10,minimum size=3.2mm,inner sep=0pt}]
 \node at (1.05,1.75){\scriptsize left block $\cong\R$ : component $u$};
 \foreach \i in {0,1,2,3}{\node[q] (L\i) at (\i*0.7,1.1){};}
 \node at (5.0,1.75){\scriptsize right block $\cong\R$ : component $v$};
 \foreach \i in {0,1,2,3}{\node[q] (R\i) at (4.0+\i*0.7,1.1){};}
 \node[cc] (chk) at (2.45,-0.15){};
 \node at (2.45,-0.55){\scriptsize one $X$-check reads $au+bv$};
 \foreach \i in {0,1,2}{\draw (chk)--(L\i);}
 \foreach \i in {1,2,3}{\draw[densely dashed] (chk)--(R\i);}
 \node at (0.00,0.55){\scriptsize multiply by $a$};
 \node at (5.9,0.55){\scriptsize multiply by $b$};
\end{tikzpicture}
\caption{Block structure of a BB code. A $Z$-logical operator is a pair $(u,v)$ with $u$ on the left block and $v$ on the right. An $X$-check multiplies the left component by $a$ and the right by $b$ and sums the two.}
\label{fig:struct}
\end{figure}
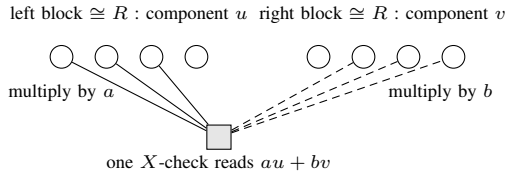

\begin{example}[parity-check rows for $\code{18}{4}{4}$]\label{ex:hx}
On $\ell=m=3$ the index map $x^iy^j\mapsto 3i+j$ turns $\R$ into $\F_2^9$.
For this support picture we use the row-vector convention of Remark~\ref{rem:conv}: check at torus site $(i,j)$ involves left-block qubit $(i+p,j+q)$ for $(p,q)\in\supp(a)$ and right-block qubit $(i+p,j+q)$ for $(p,q)\in\supp(b)$. In the column-vector multiplication convention used for the algebra, the corresponding matrix rows carry the reciprocal supports; the two descriptions are related by the involution of Remark~\ref{rem:conv}.
Under this support convention, the displayed row for check $(0,0)$ is
\begin{equation}
 h_{(0,0)}=\bigl[\;\underbrace{0\;1\;0\;\;1\;1\;0\;\;0\;0\;0}_{\text{left: }\{1,3,4\}=\supp(a)}\;\big|\;
             \underbrace{0\;0\;1\;\;0\;0\;0\;\;1\;0\;1}_{\text{right: }\{11,15,17\}=\supp(b)+9}\;\bigr],
\end{equation}
coupling left-block qubits $\{y,x,xy\}$ (indices $\{1,3,4\}$, the support of $a$) and right-block qubits $\{y^2,x^2,x^2y^2\}$ (indices $\{11,15,17\}$, the support of $b$).

The matrix $H_X=[A\mid B]$ is \emph{block circulant with circulant blocks} (BCCB): check $(i,j)$ is the $2$-d torus-translate of check $(0,0)$ by the group element $(i,j)$.
This is not a simple cyclic shift of the length-$18$ row.
For example, check $(0,1)$ shifts every qubit index by $(0,1)$ on the torus, turning the left-block pattern $\{y,x,xy\}$ into $\{y^2,xy,xy^2\}=\{2,4,5\}$, which here happens to agree with the flat $+1$ shift. The two part company as soon as the $y$-index wraps: check $(0,2)$ sends $y\mapsto y^3=1$ and $xy\mapsto xy^3=x$, giving the pattern $\{1,x,xy^2\}=\{0,3,5\}$, whereas a flat $+2$ shift of $\{1,3,4\}$ would give $\{3,5,6\}$.
Fig.~\ref{fig:bccb} shows checks $(0,0)$ and $(1,0)$ side by side to illustrate.
\end{example}

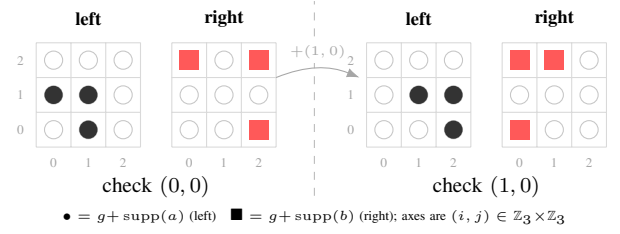
\begin{figure}[!t]
\centering
\begin{tikzpicture}[scale=0.46,font=\scriptsize,
  dot/.style={circle,fill=black!80,minimum size=2.4mm,inner sep=0pt},
  sq/.style={fill=red!65,minimum size=2.6mm,inner sep=0pt},
  od/.style={circle,draw=gray!45,minimum size=2.4mm,inner sep=0pt},
  lb/.style={font=\scriptsize\bfseries},
  axl/.style={font=\tiny,gray!70}]

\begin{scope}[shift={(0,0)}]
  \draw[step=1,gray!30,thin] (0,0) grid (3,3);
  \foreach \i in {0,1,2}\foreach \j in {0,1,2}{\node[od] at (\i+.5,\j+.5){};}
  \node[dot] at (1.5,0.5){}; \node[dot] at (0.5,1.5){}; \node[dot] at (1.5,1.5){};
  \node[lb] at (1.5,3.7){left};
  \foreach \i/\t in {0/0,1/1,2/2}{\node[axl] at (\i+.5,-0.45){\t};}
  \foreach \j/\t in {0/0,1/1,2/2}{\node[axl] at (-0.45,\j+.5){\t};}
\end{scope}

\begin{scope}[shift={(3.9,0)}]
  \draw[step=1,gray!30,thin] (0,0) grid (3,3);
  \foreach \i in {0,1,2}\foreach \j in {0,1,2}{\node[od] at (\i+.5,\j+.5){};}
  \node[sq] at (2.5,0.5){}; \node[sq] at (0.5,2.5){}; \node[sq] at (2.5,2.5){};
  \node[lb] at (1.5,3.7){right};
  \foreach \i/\t in {0/0,1/1,2/2}{\node[axl] at (\i+.5,-0.45){\t};}
\end{scope}
\node[font=\footnotesize] at (3.45,-1.15){check $(0,0)$};

\draw[gray!55,dashed] (8.0,-1.4)--(8.0,4.2);

\draw[-{Latex},gray!70] (6.9,2.0) to[bend left=22] 
  node[above=2pt,font=\tiny,gray!70,fill=white,inner sep=1.5pt]{$+(1,0)$} (9.3,2.0);

\begin{scope}[shift={(9.5,0)}]
  \draw[step=1,gray!30,thin] (0,0) grid (3,3);
  \foreach \i in {0,1,2}\foreach \j in {0,1,2}{\node[od] at (\i+.5,\j+.5){};}
  \node[dot] at (2.5,0.5){}; \node[dot] at (1.5,1.5){}; \node[dot] at (2.5,1.5){};
  \node[lb] at (1.5,3.7){left};
  \foreach \i/\t in {0/0,1/1,2/2}{\node[axl] at (\i+.5,-0.45){\t};}
  \foreach \j/\t in {0/0,1/1,2/2}{\node[axl] at (-0.45,\j+.5){\t};}
\end{scope}

\begin{scope}[shift={(13.4,0)}]
  \draw[step=1,gray!30,thin] (0,0) grid (3,3);
  \foreach \i in {0,1,2}\foreach \j in {0,1,2}{\node[od] at (\i+.5,\j+.5){};}
  \node[sq] at (0.5,0.5){}; \node[sq] at (1.5,2.5){}; \node[sq] at (0.5,2.5){};
  \node[lb] at (1.5,3.7){right};
  \foreach \i/\t in {0/0,1/1,2/2}{\node[axl] at (\i+.5,-0.45){\t};}
\end{scope}
\node[font=\footnotesize] at (12.95,-1.15){check $(1,0)$};

\node[font=\tiny] at (8.0,-2.0){$\bullet=g{+}\supp(a)$ (left)\quad $\blacksquare=g{+}\supp(b)$ (right); axes are $(i,j)\in\mathbb Z_3{\times}\mathbb Z_3$};
\end{tikzpicture}
\caption{Checks $(0,0)$ and $(1,0)$ for $\code{18}{4}{4}$. Check $g$ involves left-block qubits at $g{+}\supp(a)$ ($\bullet$) and right-block qubits at $g{+}\supp(b)$ ($\blacksquare$), indices reduced on $\mathbb{Z}_3{\times}\mathbb{Z}_3$. The pattern shifts by a torus translation, not a flat row shift; this is the BCCB structure. The support picture uses the reciprocal row convention of Remark~\ref{rem:conv}.}
\label{fig:bccb}
\end{figure}

The $Z$-logical operators are the vectors of $\ker\HX$ that are not $Z$-stabilizers. In polynomial form,
\begin{equation}
K=\{(u,v)\in\R^2:au+bv=0\},\quad S=\{(br,ar):r\in\R\},
\end{equation}
so $K$ is the syzygy module of $(a,b)$, $S$ the stabilizer submodule, and the logical operators form the quotient $K/S$ of dimension $k$. In coding terms: $K$ collects the $Z$-type operators that commute with every $X$-check, $S$ collects those that are products of stabilizers and therefore act trivially, and $K/S$ is what remains, the nontrivial $Z$-logical operators, with two operators identified when they differ by a stabilizer. The minimum distance is the lightest nonzero element of that quotient, which is why the paper works with $K/S$ rather than with $K$ alone.

\begin{remark}[stabilizer convention]\label{rem:conv}
We treat codewords as column vectors, so the $Z$-stabilizers are the columns of $\HZ^{\mathsf T}=[\,B;A\,]$, namely $(Br,Ar)=(br,ar)$, giving $S=\{(br,ar):r\in\R\}$. Under the alternative row-vector convention the stabilizers are $(\bar b r,\bar a r)$ with the reciprocals $\bar a(x,y)=a(x^{-1},y^{-1})$, $\bar b(x,y)=b(x^{-1},y^{-1})$. Since $x\mapsto x^{-1},\,y\mapsto y^{-1}$ is a ring automorphism of $\R$, it carries annihilators to annihilators and colon ideals to colon ideals, so the decomposition below is unchanged; only the labels are reciprocated.
\end{remark}

The minimum $Z$-distance is
\begin{equation}\label{eq:dz}
d_Z=\min\{\wt(u)+\wt(v):(u,v)\in K\setminus S\},
\end{equation}
and the code distance is $d=\min(d_X,d_Z)$, with $d_X$ given by the same expressions after exchanging the roles of the two blocks. The number of logical qubits follows from $\ell m$ and $\operatorname{rank}\HX$, and in closed form from a gcd when $\gcd(\ell,m)=1$ \cite{postema2026}; we use the rank, which needs no coprimality.

\subsection{Algebraic and matrix dictionary}\label{sec:matrix}
The decomposition in Section~\ref{sec:decomp} is built from a few standard constructions in $\R$ \cite{macwilliams1977,coxlittleoshea}, which we fix here together with their coordinate form; each is a finite-dimensional $\F_2$-vector space computed by Gaussian elimination. Writing $N=\ell m$ and identifying $\R$ with $\F_2^N$ by monomial coefficients, multiplication by $f$ is the matrix $F$, so $a,b$ become $A,B$. The \emph{principal ideal} $(a)=\{ar:r\in\R\}$ is the column space $\operatorname{im}A$, equivalently the abelian code generated by $a$; the membership test $bv\in(a)$ separates the two families of logical operators below. The \emph{annihilator} $\ann(a)=\{r\in\R:ar=0\}=\ker A$ is nonzero exactly when $a$ is a zero divisor in $\R$, which is common when $\ell$ or $m$ is even (then $x^\ell-1$ or $y^m-1$ is not squarefree over $\F_2$) but also occurs in the semisimple case when $a$ vanishes on a nontrivial set of characters. As a repeated-root example, $(1+x)^2=0$ in $\F_2[x]/(x^2-1)$. An element $u\in\ann(a)$ gives a one-sided logical $(u,0)$. The \emph{colon ideal}
\begin{equation}\label{eq:colon}
(a\cln b)=\{v\in\R:bv\in(a)\}=B^{-1}(\operatorname{im}A)
\end{equation}
collects the right-block words $v$ whose contribution $bv$ is cancellable by a left-block word, $au=bv$; it contains both $(a)$ and $\ann(b)$. Finally, the syzygy and stabilizer modules are
\begin{equation}\label{eq:KSmat}
K=\ker[\,A\mid B\,]=\ker\HX,\qquad S=\operatorname{im}\begin{bmatrix}B\\ A\end{bmatrix}=\operatorname{im}\HZ^{\mathsf T},
\end{equation}
so the $Z$-logical operators are the quotient $K/S=\ker\HX/\operatorname{im}\HZ^{\mathsf T}$ of dimension $k=\dim\ker\HX-\operatorname{rank}\HZ$. In computation the colon ideal is realized as $B^{-1}(\operatorname{im}A)=\ker(C_AB)$, where $C_A$ is any parity-check matrix for the code $(a)$, that is $\ker C_A=\operatorname{im}A$. Table~\ref{tab:glance} pairs each object with its coordinate form; the distance quantities built from them appear with the distance theorem in Section~\ref{sec:decomp}.

\begin{table}[!t]
\centering\small
\caption{Each object in its algebraic and coordinate forms, under $\R\cong\F_2^N$ ($N=\ell m$) with $f\leftrightarrow F$.}
\label{tab:glance}
\renewcommand{\arraystretch}{1.25}
\begin{tabular}{@{}>{\raggedright\arraybackslash}p{1.4cm}>{\raggedright\arraybackslash}p{3.35cm}>{\raggedright\arraybackslash}p{2.5cm}@{}}
\toprule
object & algebraic meaning & matrix form \\
\midrule
$(a)$ & multiples of $a$ & $\operatorname{im}A$ \\
$\ann(a)$ & $ar=0$ & $\ker A$ \\
$b\,\ann(a)$ & one-sided reps that are stabilizers & $B(\ker A)$ \\
$(a\cln b)$ & $bv\in(a)$ & $B^{-1}(\operatorname{im}A)$ \\
$K$ & syzygies $au+bv=0$ & $\ker[\,A\mid B\,]$ \\
$S$ & $Z$-stabilizers & $\operatorname{im}[\,B;A\,]$ \\
$K/S$ & logical classes, $\dim=k$ & $\ker[\,A\mid B\,]/\operatorname{im}[\,B;A\,]$ \\
\bottomrule
\end{tabular}
\end{table}

\section{Decomposition of the logical quotient}\label{sec:decomp}
The central structural result is a short exact sequence for $K/S$. Its kernel is the annihilator-derived subspace; its quotient is the colon quotient. The projection $\pi$ determines whether a logical class lies in the annihilator component or has a nonzero image in the colon quotient, with no direct-sum splitting required and no BB code excluded.

\begin{theorem}[annihilator--colon exact sequence]\label{thm:exact}
With $K,S$ as above and $(a\cln b)=\{v:bv\in(a)\}$, there is a short exact sequence of $\F_2$-vector spaces
\begin{equation}\label{eq:ses}
0\longrightarrow \ann(a)/b\,\ann(a)\xrightarrow{\ \iota\ } K/S \xrightarrow{\ \pi\ } (a\cln b)/(a)\longrightarrow0,
\end{equation}
where $\iota(t+b\ann(a))=[(t,0)]$ and $\pi[(u,v)]=v+(a)$. Hence a logical class lies in the annihilator component precisely when $\pi=0$, equivalently when it admits a representative $(t,0)$ with $at=0$; otherwise it lies in the colon component. These are labels of logical classes under the chosen projection, not labels of the support shape of their minimum-weight representatives.
\end{theorem}

\begin{proof}
Set $\pi[(u,v)]=v+(a)$. This is well defined: if $au+bv=0$ then $bv=au\in(a)$, so $v\in(a\cln b)$, and shifting $(u,v)$ by a stabilizer $(br,ar)$ changes $v$ only by the element $ar\in(a)$, leaving $v+(a)$ unchanged. The map is onto $(a\cln b)/(a)$, since for any $v$ with $bv=au$ the pair $(u,v)$ lies in $K$ and maps to $v+(a)$.

It remains to identify the kernel. Suppose $\pi[(u,v)]=0$, so $v=ar$ for some $r$. Then $au+bv=0$ becomes $a(u+br)=0$, so $t:=u+br$ lies in $\ann(a)$, and subtracting the stabilizer $(br,ar)$ shows that $(u,v)$ is represented by $(t,0)$. Conversely, $(t,0)$ is itself a stabilizer precisely when $t\in b\,\ann(a)$. The same equivalence gives injectivity of $\iota$ for free: $\iota(t+b\,\ann(a))=0$ in $K/S$ means $(t,0)\in S$, i.e.\ $t\in b\,\ann(a)$, i.e.\ the class $t+b\,\ann(a)$ was already zero. Hence $\ker\pi=\operatorname{im}\iota$, and the sequence is exact.
\end{proof}

In the dictionary of Section~\ref{sec:matrix}, the sequence \eqref{eq:ses} reads
\begin{multline}\label{eq:sesmat}
0\to\ker A/B(\ker A)\to\ker[\,A\mid B\,]\big/\operatorname{im}[\,B;A\,]\\
\to B^{-1}(\operatorname{im}A)/\operatorname{im}A\to0,
\end{multline}
so the annihilator part is the nullspace of $A$ modulo its image under $B$, and the colon part is the preimage of $\operatorname{im}A$ under $B$ modulo $\operatorname{im}A$. Every map is $\F_2$-linear and each quotient is computed by Gaussian elimination.

\begin{remark}[relation to standard constructions]\label{rem:standard}
The exact sequence \eqref{eq:ses} is a standard relation-module construction in commutative algebra~\cite{eisenbud1995}, specialized to the rank-one relation $au+bv=0$ over $\R$. The contribution here is not that construction itself but its interpretation for the BB logical quotient, and its use in separating the distance calculation into the classes introduced below.
\end{remark}

The map $\pi[(u,v)]=v+(a)$ is invariant under stabilizer shifts because a stabilizer changes $v$ by an element of $(a)$. Thus ``annihilator'' and ``colon'' refer to a class in $K/S$. A representative may nevertheless change shape after adding a stabilizer; Section~\ref{sec:tax} keeps these notions separate.

\subsection{Basis, dimension, and symmetry}
The sequence also gives a basis for $K/S$ and the corresponding dimension formula. There is an analogous sequence with $a$ and $b$ interchanged. The construction is equivariant under translation. None of this needs more than linear algebra over $\F_2$.

\begin{corollary}[logical basis and dimension]\label{cor:basis}
Let $t_1,\dots,t_{r_{\mathcal A}}$ lift a basis of $\ann(a)/b\,\ann(a)$, and let $v_1,\dots,v_{r_{\mathcal C}}$ lift a basis of $(a\cln b)/(a)$ with $u_j$ any solution of $au_j=bv_j$. Then
\begin{equation}
\{(t_i,0)\}_{i=1}^{r_{\mathcal A}}\ \cup\ \{(u_j,v_j)\}_{j=1}^{r_{\mathcal C}}
\end{equation}
is a basis of $K/S$, and
\begin{equation}\label{eq:dimk}
k=\bigl(\dim\ann(a)-\dim b\,\ann(a)\bigr)+\bigl(\dim(a\cln b)-\dim(a)\bigr).
\end{equation}
\end{corollary}
\begin{proof}
Dimension is additive in short exact sequences, which gives \eqref{eq:dimk} directly. A basis of the extension is obtained from a basis of the kernel together with lifts of a basis of the quotient; here the kernel basis is $(t_i,0)$ and the lifts are the colon witnesses $(u_j,v_j)$.
\end{proof}

\noindent\textbf{Matrix realisation.}  In the language of matrices (Section~\ref{sec:matrix}), the basis of Corollary~\ref{cor:basis} is computed as follows. The annihilator basis vectors $t_i$ are a basis of $\ker A$ reduced modulo $B(\ker A)$: one takes a basis of $\ker A$ by RREF on $A$, forms $B(\ker A)$ by applying $B$, and reduces again to extract $r_{\mathcal A}$ independent classes. The colon basis vectors $v_j$ come from $\ker(C_A B)$ (where $C_A$ is any parity-check matrix for $\operatorname{im}A$), reduced modulo $\operatorname{im}A$; each $u_j$ is a particular solution of $Au_j=Bv_j$, obtained by back-substitution. The entire construction requires only Gaussian elimination on matrices of size at most $N\times N$ with $N=\ell m$.

\begin{example}[exact-sequence trace for $\code{18}{4}{4}$]\label{ex:seq}
We trace each term of \eqref{eq:ses} for $\ell=m=3$, $a=x+y+xy$, $b=a^2$, by Gaussian elimination on $9\times9$ matrices.
\begin{enumerate}[leftmargin=1.4em,itemsep=2pt]
\item $\ann(a)=\ker A$ has dimension $2$.  Two independent annihilators are
\begin{align*}
t_1&=x+xy^2+x^2+x^2y+y+y^2,\\
t_2&=1+x^2+x^2y^2+xy+xy^2+y,
\end{align*}
both confirmed by $At_1=At_2=\mathbf{0}$.
\item $b\,\ann(a)=B(\ker A)$ has rank $0$: $Bt_1=Bt_2=\mathbf{0}$, so every annihilator maps to zero under $b$ and $\iota$ embeds $\ann(a)$ entirely into $K/S$.  Hence $r_{\mathcal A}=\dim\ann(a)-0=2$.
\item $\operatorname{rank}A=7$, so $\dim(a)=7$; the colon ideal $(a\cln b)$ equals all of $\R$ (dimension $9$), giving $r_{\mathcal C}=9-7=2$.
\end{enumerate}
Adding up, $k=r_{\mathcal A}+r_{\mathcal C}=4$. The annihilator quotient contributes two basis classes that admit one-sided representatives, while the colon quotient contributes two independent quotient coordinates. Minimum-weight representatives need not preserve those shapes; Example~\ref{ex:18} below gives the counterexample.
\end{example}
\noindent Identity \eqref{eq:dimk} is verified for every code in Table~\ref{tab:profile}. Lemma~\ref{lem:frob} shows that the equality $r_{\mathcal A}=r_{\mathcal C}=k/2$ is in fact general, not a peculiarity of the examples. This basis construction is independent of any distance search and gives an explicit algebraic generating set for all $Z$-logical classes; it complements logical-basis constructions aimed at transversal gates \cite{eberhardt2024} by organizing the basis according to the annihilator/colon classes.

\begin{lemma}[Frobenius balance]\label{lem:frob}
For every $a,b\in\R$,
\begin{equation}\label{eq:frobcolon}
(a\cln b)=\ann\!\bigl(b\,\ann(a)\bigr).
\end{equation}
Consequently
\begin{equation}\label{eq:balanced-dim}
r_{\mathcal C}=r_{\mathcal A},\qquad
k=2r_{\mathcal A}=2\bigl(\dim\ann(a)-\dim b\ann(a)\bigr).
\end{equation}
\end{lemma}
\begin{proof}
The finite group algebra $\R=\F_2[\mathbb Z_\ell\times\mathbb Z_m]$ is Frobenius. Use the nondegenerate pairing $\langle f,g\rangle=[1](fg)$, where $[1](fg)$ is the coefficient of the identity monomial. Let $I$ be an ideal and $g\in I^\perp$. For any $f\in I$ and any group monomial $m_h=x^iy^j$, the coefficient of $m_h$ in $gf$ satisfies
\[
 [m_h](gf)=\langle g,m_h^{-1}f\rangle=0,
\]
because $m_h^{-1}f\in I$. Thus $gf=0$ for every $f\in I$, so $g\in\ann(I)$. The converse follows immediately from the pairing, hence $\ann(I)=I^\perp$. Therefore $\dim\ann(I)=N-\dim I$ and $\ann(\ann(I))=I$.

Now $v\in(a\cln b)$ iff $bv\in(a)=\ann(\ann(a))$, which is equivalent to $(bv)t=0$ for every $t\in\ann(a)$, or $v\in\ann(b\ann(a))$. Therefore
\begin{align*}
\dim(a\cln b)&=N-\dim b\ann(a),\\
\dim(a)&=N-\dim\ann(a).
\end{align*}
Subtracting gives $r_{\mathcal C}=r_{\mathcal A}$. Corollary~\ref{cor:basis} then gives \eqref{eq:balanced-dim}.
\end{proof}

\begin{corollary}[symmetric sequence]\label{cor:sym}
Interchanging the two blocks gives the companion exact sequence
\begin{equation}
0\to\ann(b)/a\,\ann(b)\to K/S\to(b\cln a)/(b)\to0,
\end{equation}
by applying Theorem~\ref{thm:exact} with the roles of $a$ and $b$ exchanged. The two projections describe the same quotient in different coordinates: a class may be annihilator-type for one projection and colon-type for the other. Representative shape is a separate question.
\end{corollary}
\begin{proof}
Exchange $a$ and $b$ throughout. Then $H_X=[A\mid B]$ is symmetric in $A$ and $B$, and the relation $bu'+av'=0$ defines the right-block analogue of the syzygies. The maps $\iota'(s+a\,\ann(b))=[(0,s)]$ and $\pi'[(u,v)]=u+(b)$ satisfy the analogous exactness conditions by the same argument as Theorem~\ref{thm:exact}. A class lies in the annihilator family of the right projection iff $u\in(b)$; the two projections use different coordinate bases but identify the same quotient $K/S$.
\end{proof}
\begin{proposition}[translation equivariance]\label{prop:equiv}
All maps in \eqref{eq:ses} are $\R$-linear, hence equivariant under the translation group $G=\mathbb Z_\ell\times\mathbb Z_m$ acting by multiplication by monomials. Consequently $G$ acts on $K/S$, on $\ann(a)/b\,\ann(a)$, and on $(a\cln b)/(a)$, and the annihilator and colon families are unions of $G$-orbits.
\end{proposition}
\begin{proof}
Each map in \eqref{eq:ses} is defined by multiplication and ideal membership, both of which commute with the action of $G$ (left multiplication by a monomial). Thus $\iota$ and $\pi$ intertwine the $G$-actions on the three modules, and the images and kernels are $G$-stable subspaces.
\end{proof}
\noindent This equivariance is what makes the orbit-reduced search of Section~\ref{sec:lb} consistent: translating a witness or a support stays within the same family.

\subsection{Componentwise coset-leader formulation}
The projection $\pi$ also gives a convenient way to set up the distance problem. A logical class belongs to the annihilator component when its second coordinate lies in $(a)$; otherwise it belongs to the colon component. The weight can then be minimized separately over the two sets of classes,
\begin{equation}
\mathcal A=\{[(u,v)]\ne0:v\in(a)\},\quad
\mathcal C=\{[(u,v)]:v\notin(a)\},
\end{equation}
the annihilator and colon classes. For the colon side, fix $v\in(a\cln b)\setminus(a)$. The equation $au=bv$ determines an affine space of possible $u$'s, namely a coset $u_0+\ann(a)$; write
\begin{equation}\label{eq:lambda}
\lambda_a(w)=\begin{cases}\min\{\wt(u):au=w\} & w\in(a),\\[2pt] +\infty & \text{otherwise,}\end{cases}
\end{equation}
for the coset-leader weight of that space. With this convention the colon minimum below may equivalently be taken over all $v\in\R\setminus(a)$, since $\lambda_a(bv)=+\infty$ unless $bv\in(a)$, i.e.\ unless $v\in(a\cln b)$.

\begin{theorem}[componentwise distance formula and one-sided witness bound]\label{thm:ddec}
With $\mathcal A,\mathcal C$ as above,
\begin{equation}\label{eq:ddec}
d_Z=\min(d_{\mathcal A},d_{\mathcal C}),
\end{equation}
where the two components are given exactly by
\begin{align}
d_{\mathcal A}&=\min_{\substack{t\in\ann(a)\setminus b\,\ann(a)\\ r\in\R}}\bigl(\wt(t+br)+\wt(ar)\bigr),\label{eq:dann}\\
d_{\mathcal C}&=\min_{\,v\in(a\cln b)\setminus(a)}\bigl(\wt(v)+\lambda_a(bv)\bigr).\label{eq:dcol}
\end{align}
Setting $r=0$ in \eqref{eq:dann} gives the one-sided witness bound
\begin{equation}\label{eq:wann}
d_{\mathcal A}\le w_{\ann}:=\min\{\wt(t):t\in\ann(a)\setminus b\,\ann(a)\}.
\end{equation}
The spaces $\ann(a)$ and $b\,\ann(a)$ are found by Gaussian elimination in $O(N^3)$ time, but the minimum weight $w_{\ann}$ is itself a minimum-weight problem over the quotient $\ann(a)/b\,\ann(a)$ and is NP-hard in general. Any representative found by a bounded search over $\ann(a)\setminus b\,\ann(a)$ has weight at least $w_{\ann}$ and gives the upper bound $d_{\mathcal A}\le w_{\ann}^{\mathrm{found}}$.
The symmetric projection $u\mapsto u+(b)$ gives an equivalent componentwise formulation of the same quotient $K/S$; the two parameterizations may produce different explicit witnesses, and the lighter witness is the stronger upper bound, but the exact value $d_Z$ is unchanged. The code distance is $d=\min(d_X,d_Z)$, with $d_X$ treated in Section~\ref{sec:lb}.
\end{theorem}

\begin{proof}
Whether $v\in(a)$ is a class invariant, since stabilizers change $v$ by $ar\in(a)$; this gives the partition \eqref{eq:ddec}. For \eqref{eq:dann}, Theorem~\ref{thm:exact} writes every element of $\mathcal A$ as $(t+br,ar)$ with $t\in\ann(a)\setminus b\,\ann(a)$ and $r\in\R$, and its weight is $\wt(t+br)+\wt(ar)$. For \eqref{eq:dcol}, a colon class has all representatives with $v\notin(a)$, and for fixed $v$ the lightest left component solving $au=bv$ has weight $\lambda_a(bv)$; minimizing over $v$ gives the formula. Although the outer minimization in \eqref{eq:dcol} ranges over representatives $v$ rather than classes, this does not change the minimum, since replacing $v$ by $v+ar$ is exactly the effect of adding a stabilizer. Finally $r=0$ yields \eqref{eq:wann}.
\end{proof}

Figure~\ref{fig:geom} summarizes the componentwise minimum in \eqref{eq:ddec}.

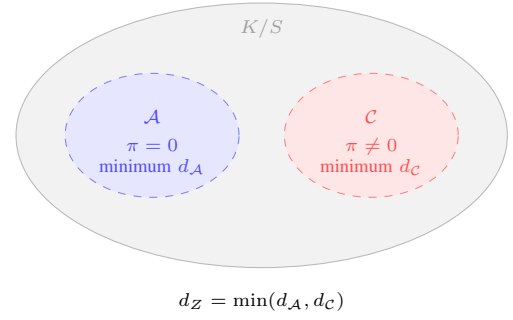
\begin{figure}[!t]
\centering
\begin{tikzpicture}[font=\scriptsize]
 \draw[fill=gray!10,draw=gray!55] (0,0) ellipse (3.25cm and 1.75cm);
 \node[gray!75] at (0,1.42) {$K/S$};

 \draw[fill=blue!10,draw=blue!50,dashed] (-1.45,0) ellipse (1.15cm and 0.82cm);
 \node[blue!70] at (-1.45,0.25) {$\mathcal A$};
 \node[blue!70,align=center] at (-1.45,-0.28) {$\pi=0$\\minimum $d_{\mathcal A}$};

 \draw[fill=red!10,draw=red!50,dashed] (1.45,0) ellipse (1.15cm and 0.82cm);
 \node[red!70] at (1.45,0.25) {$\mathcal C$};
 \node[red!70,align=center] at (1.45,-0.28) {$\pi\ne0$\\minimum $d_{\mathcal C}$};

 \node[below] at (0,-1.95) {$d_Z=\min(d_{\mathcal A},d_{\mathcal C})$};
\end{tikzpicture}
\caption{Schematic of the componentwise distance decomposition. The projection $\pi$ partitions the nonzero logical classes into the annihilator component $\mathcal A$ and the colon component $\mathcal C$. The minimum weight is computed within each component and the smaller value gives $d_Z$. Representative shape is a separate label, introduced in Section~\ref{sec:tax}.}
\label{fig:geom}
\end{figure}

\noindent\textbf{Matrix form.}  The two minimisations in the matrix dictionary are those of \eqref{eq:dannmat}--\eqref{eq:dcolmat}, which we record here for reference. The minimisation for $d_{\mathcal A}$ asks for the lightest nonzero vector of the form $(t{+}Br,\,Ar)$ with $t\in\ker A\setminus B(\ker A)$; geometrically, $Ar$ parametrises the right block and $Br$ the corresponding left-block correction. This is a minimum-weight search over a coset of the binary code $\operatorname{im}[B\,;\,A]$ inside $\ker A\oplus\F_2^N$, which is exactly the stabilizer space restricted to annihilator-type codewords. The minimisation for $d_{\mathcal C}$ treats $v$ as a right-block seed and $\lambda_A(Bv)$ as the cheapest left partner; both components are computed in $\F_2^N$ by linear algebra plus a weight minimisation. The structural dimensions in Table~\ref{tab:profile} come from these matrix operations; the weight columns additionally require the minimum-weight searches stated explicitly above.

The algebraic part of Equations \eqref{eq:dann}--\eqref{eq:dcol} is exact: the algebra fixes the structure, namely the cosets of $\ann(a)$ and the colon ideal $(a\cln b)$. The remaining computation is a minimum-weight search within those cosets, the coset-leader weight $\lambda_a$ and the minimizations over it. In the matrix dictionary of Section~\ref{sec:matrix} these read
\begin{align}
d_{\mathcal A}&=\min_{\substack{t\in\ker A\setminus B(\ker A)\\ r\in\F_2^N}}\wt\begin{bmatrix}t+Br\\ Ar\end{bmatrix},\label{eq:dannmat}\\
d_{\mathcal C}&=\min_{v\in B^{-1}(\operatorname{im}A)\setminus\operatorname{im}A}\bigl(\wt(v)+\lambda_A(Bv)\bigr),\label{eq:dcolmat}
\end{align}
with $\lambda_A(w)=\min\{\wt(u):Au=w\}$, so each component is a coset-leader weight of a binary linear code.

\begin{remark}[what is exact and what remains a search]\label{rem:exact}
Equations~\eqref{eq:dann}--\eqref{eq:dcol} are exact identities, but neither is a closed-form distance formula. The residual minimizations are coset-leader problems for binary linear codes and are NP-hard in general~\cite{vardy1997}. In particular, the one-sided quantity
\[
w_{\ann}=\min\{\wt(t):t\in\ann(a)\setminus b\ann(a)\}
\]
is only an upper bound on $d_{\mathcal A}$: adding a stabilizer can make the same annihilator class lighter while populating the other block. The $\code{18}{4}{4}$ example has $w_{\ann}=6$ but $d_{\mathcal A}=4$.
\end{remark}

\subsection{Spectral interpretation in the semisimple case}\label{sec:spectral}
When $\ell$ and $m$ are odd, $x^\ell-1$ and $y^m-1$ are separable and $\R$ is semisimple. Over a splitting field $\mathbb F_q\supseteq\F_2$ containing all $\ell$th and $m$th roots of unity, $\R\otimes\mathbb F_q\cong\prod_\chi\mathbb F_q$ indexed by characters $\chi=(\alpha,\beta)$ with $\alpha^\ell=\beta^m=1$, and multiplication by $f$ acts on the $\chi$-component as the scalar $f(\chi)=f(\alpha,\beta)$. The syzygy condition $au+bv=0$ becomes, fibrewise, $a(\chi)\hat u_\chi+b(\chi)\hat v_\chi=0$.

\begin{proposition}[semisimple fiber form]\label{prop:spec}
Let $\ell,m$ be odd. A character $\chi$ contributes to $K/S$ if and only if $a(\chi)=b(\chi)=0$, and each such common zero contributes a two-dimensional fiber. Consequently
\begin{equation}\label{eq:speck}
k=2\,\bigl|\{\chi:a(\chi)=b(\chi)=0\}\bigr|,
\end{equation}
and the annihilator and colon parts each account for one dimension per common zero, so $r_{\mathcal A}=r_{\mathcal C}=|\{\chi:a(\chi)=b(\chi)=0\}|=k/2$.
\end{proposition}
\begin{proof}
Fix a character $\chi$. The solution space of $a(\chi)\hat u_\chi+b(\chi)\hat v_\chi=0$ in $\mathbb F_q^2$ is two-dimensional when $a(\chi)=b(\chi)=0$ and one-dimensional otherwise. The stabilizer fiber is spanned by $(b(\chi),a(\chi))$, nonzero exactly when $(a(\chi),b(\chi))\ne(0,0)$; there it fills the one-dimensional solution space, so $K_\chi/S_\chi=0$. When $a(\chi)=b(\chi)=0$ the stabilizer fiber is zero and $K_\chi/S_\chi$ is two-dimensional, giving \eqref{eq:speck}. On the annihilator side $\ann(a)_\chi\ne0$ iff $a(\chi)=0$, and $(b\,\ann(a))_\chi$ fills it unless $b(\chi)=0$ as well, so $\ann(a)/b\,\ann(a)$ has one dimension exactly at the common zeros; the colon side is symmetric. Since extension of scalars to the splitting field is flat, the $\F_2$-dimensions are recovered by summing over the Frobenius-stable character fibres.
\end{proof}

Equation \eqref{eq:speck} recovers the common-zero (gcd) dimension count \cite{linpryadko2024,postema2026} from the quotient side. For $\code{18}{4}{4}$ the common zeros are the Frobenius-conjugate pair $(\omega,\omega^2),(\omega^2,\omega)$ over $\F_4$, giving $k=4$; for $\code{90}{8}{10}$ there are four common zeros over $\F_{16}$, giving $k=8$, both confirmed by direct evaluation. Over $\F_2$ the set of common zeros is Frobenius-closed, and the binary dimension equals the total size of the corresponding Frobenius orbits, equivalently the number of common-zero characters over the splitting field. In the semisimple case the spectral form recovers the general Frobenius balance $r_{\mathcal A}=r_{\mathcal C}=k/2$ character by character. 
\begin{example}[common zeros for $\code{18}{4}{4}$]\label{ex:spectral}
With $\ell=m=3$ odd, $\R$ is semisimple. The splitting field is $\F_4=\F_2[\omega]$ where $\omega^2+\omega+1=0$.
The characters are $\chi=(\alpha,\beta)$ with $\alpha^3=\beta^3=1$, i.e.\ $\alpha,\beta\in\{1,\omega,\omega^2\}$.
Since $b=a^2$, the condition $b(\chi)=0$ is equivalent to $a(\chi)=0$, so the
common-zero locus is
\[
a(\chi)=\alpha+\beta+\alpha\beta=0.
\]
Checking all nine pairs, only $(\omega,\omega^2)$ and $(\omega^2,\omega)$ satisfy this
(e.g.\ $\omega+\omega^2+\omega\cdot\omega^2=\omega+\omega^2+1=0$ in $\F_4$).
These two characters form one Frobenius orbit of size two. Each contributes a
two-dimensional fiber to $K/S$ over the splitting field (Proposition~\ref{prop:spec}),
so the orbit accounts for $2\times2=4$ binary dimensions, giving $k=4$.
Splitting the fiber between the two families, each character adds one dimension to
$r_{\mathcal A}$ and one to $r_{\mathcal C}$, which recovers $r_{\mathcal A}=r_{\mathcal C}=2$
from Example~\ref{ex:seq}.
\end{example}
For repeated-root rings the character decomposition is unavailable, but Lemma~\ref{lem:frob} gives the same dimension balance directly. The spectral description is therefore an interpretation and a dimension check, not an assumption needed by the distance analysis.

\section{Algebraic component versus representative shape}\label{sec:tax}
The exact sequence classifies logical \emph{classes}; Hamming weight is a property of their \emph{representatives}. These two levels should not be conflated. We therefore use two independent labels throughout the remainder of the paper.

\subsection{Two independent labels}
For the fixed projection $\pi[(u,v)]=v+(a)$, a nonzero logical class is in the \emph{annihilator component} $\mathcal A$ when $\pi=0$ and in the \emph{colon component} $\mathcal C$ when $\pi\ne0$. This dichotomy is algebraic and class-invariant.

For shape, let $\mathcal M([c])$ be the set of minimum-weight representatives of a logical class $[c]$. Fix $\tau\ge1$.
\begin{definition}[representative shape]\label{def:shape}
A nonzero logical class is
\begin{enumerate}[leftmargin=1.4em,itemsep=1pt]
\item \emph{one-sided} if some $(u,v)\in\mathcal M([c])$ has $u=0$ or $v=0$;
\item \emph{$\tau$-lopsided} if it is not one-sided and some $(u,v)\in\mathcal M([c])$ satisfies $1\le\min\{\wt(u),\wt(v)\}\le\tau$;
\item \emph{$\tau$-balanced} if every $(u,v)\in\mathcal M([c])$ satisfies $\min\{\wt(u),\wt(v)\}>\tau$.
\end{enumerate}
\end{definition}
These three shapes partition the nonzero logical classes once $\tau$ is fixed. The definitions do not exclude any of the three shapes in either algebraic component. The examples studied here use $\tau=3$, the maximum of $\wt(a)$ and $\wt(b)$. In particular, Example~\ref{ex:72} exhibits one-sided minimum representatives in both $\mathcal A$ and $\mathcal C$. Figure~\ref{fig:modes} illustrates the three support shapes, while Fig.~\ref{fig:twoaxes} summarizes the distinction between algebraic component and representative shape.

\begin{figure}[!t]
\centering
\begin{tikzpicture}[scale=0.72, cell/.style={draw,minimum size=2.7mm,inner sep=0pt}]
  \node at (-2.0,1.05){\scriptsize one-sided};
  \foreach \i in {0,1,2}{\node[cell,fill=black!35] at (-2.9+\i*0.3,0.4){};}
  \foreach \i in {0,1,2}{\node[cell] at (-1.5+\i*0.3,0.4){};}
  \node at (-2.0,-0.15){\scriptsize one block is zero};

  \node at (1.6,1.05){\scriptsize lopsided};
  \foreach \i in {0,1,2}{\node[cell,fill=black!35] at (0.7+\i*0.3,0.4){};}
  \node[cell,fill=black!35] at (2.1,0.4){};\node[cell] at (2.4,0.4){};\node[cell] at (2.7,0.4){};
  \node at (1.6,-0.15){\scriptsize one side has weight $\le\tau$};

  \node at (5.1,1.05){\scriptsize balanced};
  \node[cell,fill=black!35] at (4.2,0.4){};\node[cell,fill=black!35] at (4.5,0.4){};\node[cell] at (4.8,0.4){};
  \node[cell,fill=black!35] at (5.45,0.4){};\node[cell,fill=black!35] at (5.75,0.4){};\node[cell] at (6.05,0.4){};
  \node at (5.1,-0.15){\scriptsize both sides exceed $\tau$};
\end{tikzpicture}
\caption{Support shapes used in Definition~\ref{def:shape}. Filled cells denote nonzero coordinates and the gap separates the left ($u$) and right ($v$) blocks. These are properties of minimum-weight representatives, not algebraic component labels.}
\label{fig:modes}
\end{figure}
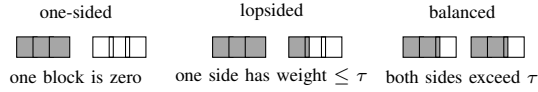

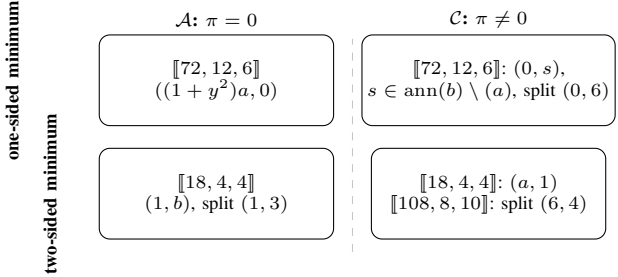
\begin{figure}[!t]
\centering
\begin{tikzpicture}[font=\scriptsize,
  cell/.style={draw,rounded corners,align=center,minimum width=31mm,minimum height=12mm,inner sep=2pt},
  hd/.style={font=\scriptsize\bfseries,align=center}]
 \node[hd] at (-1.8,1.25) {$\mathcal A$: $\pi=0$};
 \node[hd] at (1.8,1.25) {$\mathcal C$: $\pi\ne0$};
 \node[rotate=90,hd] at (-4.5,0.45) {one-sided minimum};
 \node[rotate=90,hd] at (-4.0,-1.05) {two-sided minimum};
 \node[cell] at (-1.8,0.45) {$\code{72}{12}{6}$\\$((1+y^2)a,0)$};
 \node[cell] at (1.8,0.45) {$\code{72}{12}{6}$: $(0,s)$,\\$s\in\ann(b)\setminus(a)$, split $(0,6)$};
 \node[cell] at (-1.8,-1.05) {$\code{18}{4}{4}$\\$(1,b)$, split $(1,3)$};
 \node[cell] at (1.8,-1.05) {$\code{18}{4}{4}$: $(a,1)$\\$\code{108}{8}{10}$: split $(6,4)$};
 \draw[gray!45,dashed] (0,-1.82)--(0,1.0);
\end{tikzpicture}
\caption{Algebraic component and support shape are different axes. Membership in $\mathcal A$ or $\mathcal C$ is fixed by the projection $\pi$, whereas adding a stabilizer can change the support shape of a representative. The $\code{18}{4}{4}$ code shows that an annihilator class can be lighter after populating both blocks, while $\code{72}{12}{6}$ shows that a colon class can have a one-sided minimum representative. Lopsided versus balanced further refines the two-sided row through the threshold $\tau$.}
\label{fig:twoaxes}
\end{figure}

Figure~\ref{fig:tree} summarizes the resulting hierarchy of representative shapes.

\begin{figure}[!t]
\centering
\begin{tikzpicture}[scale=0.92, every node/.style={font=\footnotesize},
  bx/.style={draw,rounded corners,inner sep=2.5pt,align=center}]
 \node[bx,fill=black!8] (top) at (0,0) {minimum-weight\\representatives $\mathcal M([c])$};
 \node[bx,fill=black!4] (one) at (-2.1,-1.3) {one-sided};
 \node[bx,fill=black!4] (two) at (2.1,-1.3) {two-sided};
 \node[bx,fill=black!3] (lop) at (0.7,-2.8) {$\tau$-lopsided};
 \node[bx,fill=black!3] (bal) at (3.3,-2.8) {$\tau$-balanced};
 \draw[-{Latex}] (top)--(one); \draw[-{Latex}] (top)--(two);
 \draw[-{Latex}] (two)--(lop); \draw[-{Latex}] (two)--(bal);
 \node[align=center,font=\scriptsize] at (-2.1,-2.35){one block\\vanishes};
 \node[align=center,font=\scriptsize] at (2.0,-3.65){weight-shape refinement\\for two-sided minima};
\end{tikzpicture}
\caption{Hierarchy of representative shapes in Definition~\ref{def:shape}. This classification is independent of the annihilator/colon component label; Fig.~\ref{fig:twoaxes} shows the two axes together.}
\label{fig:tree}
\end{figure}
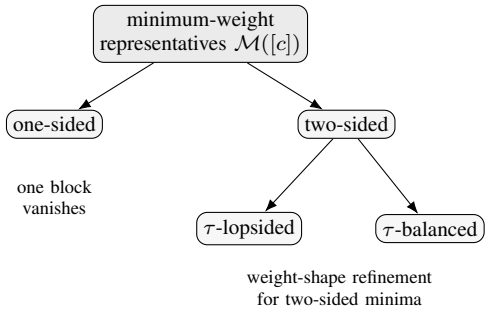

\subsection{Search order and what each step establishes}
The distinction above changes how the inexpensive screens should be interpreted. Gaussian elimination computes $\ann(a)$, $b\ann(a)$, $(a\cln b)$, and their symmetric counterparts in polynomial time. A subsequent low-weight search over those spaces is not free, and its output is an upper-bound witness and does not by itself establish a lower bound. A practical order is therefore:
\begin{enumerate}[leftmargin=1.4em,itemsep=1pt]
\item search left-only representatives $t\in\ann(a)\setminus b\ann(a)$; these are in $\mathcal A$ under $\pi$;
\item search right-only representatives $s\in\ann(b)\setminus a\ann(b)$; under $\pi$ these need not be in $\mathcal A$ and must be labelled by evaluating $\pi$;
\item search sparse two-block relations $au=bv$ and record both their component label and their weight split;
\item use a lower-bound method only for candidates that remain competitive.
\end{enumerate}
Algorithm~\ref{alg:screen} summarizes these upper-bound searches. This ordering is a heuristic for finding small upper bounds. The exact component minima remain $d_{\mathcal A}$ and $d_{\mathcal C}$ from Theorem~\ref{thm:ddec}.

\subsection{Worked examples}\label{sec:examples}

\begin{example}[one-sided witness for $\code{72}{12}{6}$]\label{ex:72}
With $\ell=m=6$, $a=x^3+y+y^2$, and $b=y^3+x+x^2$, we derive the weight-$6$ witness $u=(1+y^2)a$ from the ring arithmetic alone, without any search.

\emph{Step 1.}  Squaring is additive over $\F_2$, and $x^6=1$ in $\R$, so
\[
a^2=(x^3+y+y^2)^2=x^6+y^2+y^4=1+y^2+y^4.
\]

\emph{Step 2.}  Multiply both sides by $(1+y^2)$:
\[
\begin{aligned}(1+y^2)\cdot a^2 &= 1+y^2+y^4+y^2+y^4+y^6\\
&=1+y^6=1+1=0,\end{aligned}
\]
using $y^6=1$.  Hence $(1+y^2)a^2=0$ in $\R$.

\emph{Step 3.}  Factor as $a\cdot\bigl[(1+y^2)a\bigr]=0$, so $t=(1+y^2)a\in\ann(a)$.  Its support has six monomials: $\{x^3,y,y^2,x^3y^2,y^3,y^4\}$, so $\wt(t)=6$.

The candidate is $(u,v)=\bigl((1+y^2)a,0\bigr)$. A direct membership test gives $t\notin b\ann(a)$, so it is a nontrivial class in $\mathcal A$ and $d_{\mathcal A}\le6$. The construction is illustrated in Fig.~\ref{fig:72witness}. The cluster search of Section~\ref{sec:cluster} with radius $5$ finds no logical, so $d_Z\ge6$ and $d=6$.

The same computation with the roles of $a$ and $b$ exchanged gives $b^2=y^6+x^2+x^4=1+x^2+x^4$ and $(1+x^2)b^2=1+x^6=0$, so $s=(1+x^2)b\in\ann(b)$ with $\wt(s)=6$. Ideal membership gives $s\notin(a)$, so $(0,s)$ is a weight-$6$ logical in the colon component with a one-sided representative.
\end{example}

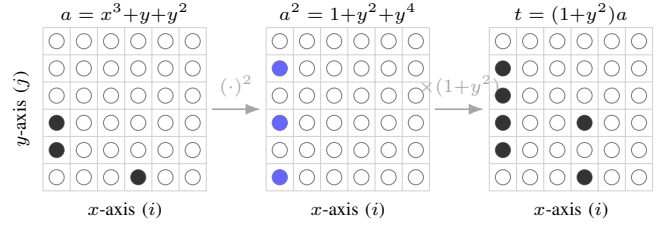
\begin{figure}[!t]
\centering
\begin{tikzpicture}[scale=0.36, font=\scriptsize,
  dot/.style={circle, fill=black!80, minimum size=2mm, inner sep=0pt},
  sq/.style={circle, fill=blue!60, minimum size=2mm, inner sep=0pt},
  od/.style={circle, draw, gray, minimum size=2mm, inner sep=0pt}]

  \begin{scope}[shift={(0,0)}]
    \draw[step=1, gray!35, thin] (0,0) grid (6,6);
    \node at (3,6.5) {$a=x^3{+}y{+}y^2$};
    \foreach \i in {0,...,5}\foreach \j in {0,...,5}{\node[od] at (\i+.5,\j+.5){};}
    
    \node[dot] at (3.5,0.5){}; 
    \node[dot] at (0.5,1.5){}; 
    \node[dot] at (0.5,2.5){}; 
    
    \node at (3,-0.7) {$x$-axis ($i$)};
    \node[rotate=90] at (-0.8,3) {$y$-axis ($j$)};
  \end{scope}

  \draw[-{Latex}, gray!70, semithick] (6.2,3) -- (8.0,3) node[midway, above=1pt]{${(\cdot)}^2$};

  \begin{scope}[shift={(8.2,0)}]
    \draw[step=1, gray!35, thin] (0,0) grid (6,6);
    \node at (3,6.5) {$a^2=1{+}y^2{+}y^4$};
    \foreach \i in {0,...,5}\foreach \j in {0,...,5}{\node[od] at (\i+.5,\j+.5){};}
    
    \node[sq] at (0.5,0.5){}; 
    \node[sq] at (0.5,2.5){}; 
    \node[sq] at (0.5,4.5){}; 
    
    \node at (3,-0.7) {$x$-axis ($i$)};
  \end{scope}

  \draw[-{Latex}, gray!70, semithick] (14.4,3) -- (16.2,3) node[midway, above=1pt]{${\times(1{+}y^2)}$};

  \begin{scope}[shift={(16.4,0)}]
    \draw[step=1, gray!35, thin] (0,0) grid (6,6);
    \node at (3,6.5) {$t=(1{+}y^2)a$};
    \foreach \i in {0,...,5}\foreach \j in {0,...,5}{\node[od] at (\i+.5,\j+.5){};}
    
    \node[dot] at (3.5,0.5){}; 
    \node[dot] at (0.5,1.5){}; 
    \node[dot] at (0.5,2.5){}; 
    \node[dot] at (3.5,2.5){}; 
    \node[dot] at (0.5,3.5){}; 
    \node[dot] at (0.5,4.5){}; 
    
    \node at (3,-0.7) {$x$-axis ($i$)};
  \end{scope}
\end{tikzpicture}
\caption{The $\code{72}{12}{6}$ one-sided witness: $(1+y^2)a^2=0$, hence $t=(1+y^2)a\in\ann(a)$ has weight $6$.}
\label{fig:72witness}
\end{figure}
\begin{example}[$\code{18}{4}{4}$: component is not shape]\label{ex:18}
Let $\ell=m=3$, $a=x+y+xy$, and $b=a^2=x^2+y^2+x^2y^2$. The colon relation $b\cdot1=a^2$ gives
\begin{equation}
(u,v)=(a,1),\qquad au+bv=0.
\end{equation}
Because $a$ is a zero divisor, it is not invertible, so $(a)$ is a proper ideal and $1\notin(a)$. Hence this weight-$4$ representative, of split $(3,1)$, lies in the colon component $\mathcal C$.

The same code also has a weight-$4$ annihilator-component logical. Since $b\in(a)$,
\begin{equation}
(u,v)=(1,b)\in K\setminus S,\qquad \pi[(1,b)]=b+(a)=0.
\end{equation}
Its class has a one-sided representative $(t,0)$ with $t=1+a^3\in\ann(a)$ and $\wt(t)=6$. Exhaustive enumeration of the three nonzero elements of $\ann(a)/b\ann(a)$ gives $w_{\ann}=6$, whereas the explicit logical $(1,b)$ shows $d_{\mathcal A}=4$:
\begin{equation}
 w_{\ann}=6,\qquad d_{\mathcal A}=4.
\end{equation}
Thus $d=4$ is attained in both algebraic components, while the displayed minimum representatives are $3$-lopsided. This is the smallest example showing why ``annihilator'' must not be used as a synonym for ``one-sided.''
\end{example}

\begin{example}[balanced minimum for $\code{108}{8}{10}$]\label{ex:108}
Here $\ell=9$, $m=6$, $a=x^3+y+y^2$, and $b=y^3+x+x^2$. The decomposition gives $r_{\mathcal A}=r_{\mathcal C}=4$. The annihilator space has $2^6=64$ elements and $b\ann(a)$ has four; exact enumeration of the remaining $60$ elements gives
\begin{equation}
 w_{\ann}=12.
\end{equation}
This is a one-sided upper bound only; it is not a lower bound on $d_{\mathcal A}$.

The search finds
\begin{align}
 u&=x+xy+x^4y^4+x^5+x^7y^3+x^8y^4,\nonumber\\
 v&=xy^4+x^3y+x^4y^2+x^6y^5,
\end{align}
with $au=bv$, split $(6,4)$, and total weight $10$. Direct ideal-membership tests give $v\notin(a)$ and $u\notin(b)$, so the witness is colon-type under both coordinate projections.

The cluster search with radius $9$ finds no logical, so $d=10$. Since the minimum nonzero weight of $K$ is $6$, Proposition~\ref{prop:irred} shows that the search with radius $11$ lists every logical of weight at most $11$. It returns only the $54$ translates of weight $10$, grouped into nine classes; every one has $\pi\neq0$ and split $(6,4)$, and no class contains a representative with a side of weight at most three. Hence $d_{\mathcal C}=10$, no annihilator logical has weight below $12$, and $w_{\ann}=12$ gives $d_{\mathcal A}=12$. All minimum-weight classes of this code are $3$-balanced colon classes.
\end{example}

\subsection{Explicit witnesses}
Table~\ref{tab:wit} lists a low-weight representative for each standard code with its weight split; all pairs are confirmed to lie in $K\setminus S$ by exact $\F_2$ computation (see Examples~\ref{ex:72}--\ref{ex:108} for the derivations of the $\code{72}{12}{6}$ and $\code{108}{8}{10}$ cases). The Gross-code pair extends the one-sided pattern of $\code{72}{12}{6}$ by an extra sparse factor $(1+x^6)$.

\begin{table*}[!t]
\centering\small
\caption{Verified minimum-weight witnesses. Every witness has weight $d$; the shapes of all minimum-weight classes are in Table~\ref{tab:census}.}
\label{tab:wit}
\renewcommand{\arraystretch}{1.25}
\begin{tabular}{@{}l>{\raggedright\arraybackslash}p{6.2cm}ccc@{}}
\toprule
code & witness & split & component & witness shape \\
\midrule
$\code{18}{4}{4}$    & $(a,1)$; also $(1,b)$ & $(3,1)$; $(1,3)$ & $\mathcal C$; $\mathcal A$ & lopsided \\
$\code{72}{12}{6}$   & $((1+y^2)a,0)$; also $(0,(1+x^2)b)$ & $(6,0)$; $(0,6)$ & $\mathcal A$; $\mathcal C$ & one-sided \\
$\code{90}{8}{10}$   & $((1{+}y)\sum_{j=0}^{4}x^{3j},0)$ & $(10,0)$ & $\mathcal A$ & one-sided \\
$\code{108}{8}{10}$  & see~\eqref{eq:w108} & $(6,4)$ & $\mathcal C$ & balanced \\
$\code{144}{12}{12}$ & $((1{+}x^6)(1{+}y^2)a,0)$ & $(12,0)$ & $\mathcal A$ & one-sided \\
$\code{288}{12}{18}$ & see~\eqref{eq:w288} & $(18,0)$ & $\mathcal A$ & one-sided \\
\bottomrule
\end{tabular}
\end{table*}

\section{Exact lower bounds}\label{sec:lb}

\noindent The componentwise formula identifies two separate searches for upper bounds, but a lower bound still requires ruling out light logicals. We first record the support-exclusion criterion as a baseline exact test; the cluster method introduced below exploits sparsity much more strongly and is the method used for the reported lower bounds. Because $K$ and $S$ are the $Z$-logical objects, the test bounds $d_Z$; the full quantum distance is recovered through Proposition~\ref{prop:sym}. A support $E$ is a subset of the $n=2\ell m$ coordinates, labelled by block and by group element $(i,j)\in\mathbb Z_\ell\times\mathbb Z_m$; the translation group $G$ acts on both blocks simultaneously. For such $E$ write $K(E)=K\cap\F_2^E$ and $S(E)=S\cap\F_2^E$, the elements supported within $E$. Since $S\subseteq K$, also $S(E)\subseteq K(E)$.

\begin{theorem}[support-exclusion criterion]\label{thm:supp}
$d_Z\ge d_0$ if and only if $K(E)=S(E)$ for every support $E$ with $|E|<d_0$.
\end{theorem}
\begin{proof}
Pick any $E$ with $|E|<d_0$ and $K(E)\ne S(E)$; a vector $c\in K(E)\setminus S(E)$ is a non-stabilizer logical of weight at most $|E|<d_0$, so $d_Z<d_0$. The converse is immediate: if $c$ is a non-stabilizer logical of weight $<d_0$ then $\supp(c)$ is an $E$ of that size with $c\in K(E)\setminus S(E)$.
\end{proof}

Two reductions cut the cost of the test. First, $K$ and $S$ are invariant under the translation group $G=\mathbb Z_\ell\times\mathbb Z_m$ acting by multiplication by monomials, and translation permutes supports, so $\dim K(E)$ and $\dim S(E)$ depend only on the orbit of $E$; one support per orbit suffices. Second, $\dim K(E)=|E|-\operatorname{rank}\HX[:,E]$, so a support whose check columns are independent passes after a single rank computation, and only the dependent supports need the comparison with $S(E)$. Algorithm~\ref{alg:supp} records the support-exclusion stage. The orbit reduction is realized cheaply by anchoring: the translation acts transitively within each block, so translating one occupied coordinate of a support to the origin of its own block gives a canonical representative. Enumerating the supports through the left-block origin, together with those lying entirely in the right block and through its origin, then visits at least one representative of every orbit without forming orbits explicitly. Explicitly, with $G_Z=\HZ^{\mathsf T}=[\,B;A\,]$ the stabilizer generator matrix and $[\,A\mid B\,]_E$ the check matrix restricted to the columns in $E$, the test is two ranks: $K(E)=\ker([\,A\mid B\,]_E)$, so $\dim K(E)=|E|-\operatorname{rank}([\,A\mid B\,]_E)$, while $S(E)=\operatorname{im}G_Z\cap\F_2^E$ has $\dim S(E)=\dim\ker\big((G_Z)_{E^c}\big)-\dim\ker G_Z$, where $(G_Z)_{E^c}$ is the submatrix of $G_Z$ formed by the rows indexed by the complement $E^c$, the coordinates a stabilizer supported in $E$ must avoid.

The same construction, applied to the $X$-logical quotient
\begin{align*}
K_X&=\{(p,q)\in\R^2:B^{\mathsf T}p+A^{\mathsf T}q=0\},\\
S_X&=\{(A^{\mathsf T}s,B^{\mathsf T}s):s\in\R\},
\end{align*}
bounds $d_X$. For the standard BB construction the two distances coincide, so a single $Z$-side run suffices.

\begin{proposition}[$X/Z$ symmetry]\label{prop:sym}
For every BB code with CSS pair $\HX=[A\mid B]$, $\HZ=[B^{\mathsf T}\mid A^{\mathsf T}]$, one has $d_X=d_Z$ (cf.~\cite{bravyi2024}). The reciprocal map $\rho:x\mapsto x^{-1},\,y\mapsto y^{-1}$ is a weight-preserving involution of $\R$ with $A^{\mathsf T}=\rho A\rho$ and $B^{\mathsf T}=\rho B\rho$, and $(p,q)\mapsto(\rho q,\rho p)$ carries the $X$-logical quotient $K_X/S_X$ onto the $Z$-logical quotient $K/S$ bijectively and weight-preservingly.
\end{proposition}
\begin{proof}
Recall that for group-algebra matrices, transposition equals multiplication by the reciprocal, $A^{\mathsf T}=\rho A\rho$ and likewise for $B$. If $B^{\mathsf T}p+A^{\mathsf T}q=0$ then $\rho B\rho\,p+\rho A\rho\,q=0$, and applying $\rho$ gives $A(\rho q)+B(\rho p)=0$, so $(\rho q,\rho p)\in K$. The map sends $S_X$, whose generators are $(A^{\mathsf T}s,B^{\mathsf T}s)$, to $(B\rho s,A\rho s)=(br,ar)\in S$ with $r=\rho s$; being an involution composed with a block swap, it preserves Hamming weight. Thus it descends to a weight-preserving isomorphism $K_X/S_X\to K/S$, and $d_X=d_Z$.
\end{proof}

Consequently a proof of $d_Z$ is also a proof of $d=\min(d_X,d_Z)=d_Z$, so the lower-bound search need only run on the $Z$ side. The proposition uses the column-vector convention of Remark~\ref{rem:conv}; under the row-vector convention the same statement holds, conjugated by $\rho$.

\subsection{Syndrome connectivity and anchored cluster search}\label{sec:cluster}
Support exclusion treats all supports alike. The sparsity of $H_X$ gives a much stronger restriction on the supports that can carry a minimum-weight logical. For a set $T$ of qubits write $\mathbf 1_T$ for its indicator vector and $\sigma(T)=H_X\mathbf 1_T$ for its syndrome.

\begin{lemma}[syndrome connectivity]\label{lem:conn}
Let $z\in K\setminus S$ have weight $d_Z$. Then $\sigma(T)\ne0$ for every nonempty $T\subsetneq\supp(z)$.
\end{lemma}
\begin{proof}
If $\sigma(T)=0$ then $\mathbf 1_T$ and $z+\mathbf 1_T$ both lie in $K$, both are nonzero, and both have weight below $d_Z$. Their sum is $z\notin S$, so at least one of them lies in $K\setminus S$, contradicting the minimality of $z$.
\end{proof}

Let $T\subsetneq\supp(z)$ be nonempty and let $g$ be any check with $\sigma(T)_g=1$. Since $z$ satisfies $g$, the set $\supp(g)\cap(\supp(z)\setminus T)$ is nonempty. Choosing $g$ canonically, as the lowest-index unsatisfied check, therefore leaves at most $\wt(a)+\wt(b)-1$ ways to extend $T$ toward $\supp(z)$. Algorithm~\ref{alg:cluster} explores these extensions from an anchor qubit fixed by translation symmetry. Write $c=\max\{\wt(a),\wt(b)\}$ for the largest number of $X$-checks meeting one qubit.

\begin{algorithm}[!t]
\caption{Anchored cluster search (exact lower bound)}\label{alg:cluster}
\begin{algorithmic}[1]
\State \textbf{Input:} $H_X$, a basis of $S$, radius $W$. \textbf{Output:} a logical of weight $\le W$, or the certificate $d_Z\ge W+1$.
\State \Call{Grow}{$\{q_{\mathrm L}\}$}, where $q_{\mathrm L}$ is the left-block origin
\State \Call{Grow}{$\{q_{\mathrm R}\}$}, where $q_{\mathrm R}$ is the right-block origin, using right-block qubits only
\State \textbf{return} $d_Z\ge W+1$
\Statex
\Function{Grow}{$T$}
  \State $u\gets\wt(\sigma(T))$
  \If{$u=0$}
    \If{$\mathbf 1_T\notin S$} \textbf{stop} and output $\mathbf 1_T$ \EndIf
    \State \textbf{return} \Comment{pruned by Lemma~\ref{lem:conn}}
  \EndIf
  \If{$|T|+\lceil u/c\rceil>W$} \textbf{return} \EndIf
  \State $g\gets$ lowest-index check with $\sigma(T)_g=1$
  \For{each qubit $q\in\supp(g)\setminus T$}
    \State \Call{Grow}{$T\cup\{q\}$}
  \EndFor
\EndFunction
\end{algorithmic}
\end{algorithm}

\begin{theorem}[exactness of the cluster search]\label{thm:cluster}
If Algorithm~\ref{alg:cluster} outputs a vector, it is an element of $K\setminus S$ of weight at most $W$. If it terminates without output, then $d_Z\ge W+1$.
\end{theorem}
\begin{proof}
Any output has zero syndrome and lies outside $S$, so it is a logical of weight $|T|\le W$. Suppose $d_Z\le W$ and let $z$ be a logical of weight $d_Z$. If $z$ meets the left block, a translation moves one of its left-block qubits to $q_{\mathrm L}$; otherwise $z$ is supported on the right block and a translation moves one of its qubits to $q_{\mathrm R}$. Translations preserve $K$, $S$, and weight, so we may assume $z$ contains the anchor of one of the two calls. Consider the chain of sets $T_1\subsetneq T_2\subsetneq\cdots\subsetneq T_{d_Z}=\supp(z)$ built from the anchor by always adding a qubit of $\supp(z)$ lying in the canonical check $g$; such a qubit exists by the remark before the algorithm, and each $T_i$ with $i<d_Z$ has $\sigma(T_i)\ne0$ by Lemma~\ref{lem:conn}. No $T_i$ with $i<d_Z$ is removed by the zero-syndrome prune. Each qubit of $\supp(z)\setminus T_i$ changes at most $c$ syndrome bits, so $d_Z-|T_i|\ge\lceil u/c\rceil$ and $T_i$ survives the parity prune. The chain therefore reaches $T_{d_Z}$, which is output unless an earlier output occurred.
\end{proof}

Letting the search continue after an output, rather than stopping, lists every logical of weight exactly $d_Z$ that contains an anchor; applying all translations then gives all minimum-weight logicals. For weights above $d_Z$ Lemma~\ref{lem:conn} no longer applies directly, but the following range suffices for the componentwise distances below.

\begin{proposition}[irreducible range]\label{prop:irred}
Let $w_K$ be the minimum weight of a nonzero element of $K$. Every nonzero $z\in K$ with $\wt(z)<2w_K$ satisfies $\sigma(T)\ne0$ for all nonempty $T\subsetneq\supp(z)$. Consequently, Algorithm~\ref{alg:cluster} with $W<2w_K$, with the output step replaced by recording and returning, lists every element of $K\setminus S$ of weight at most $W$ up to translation.
\end{proposition}
\begin{proof}
If $\sigma(T)=0$, then $\mathbf 1_T$ and $z+\mathbf 1_T$ are nonzero elements of $K$ whose weights sum to $\wt(z)<2w_K$, so one of them has weight below $w_K$, a contradiction. The completeness argument of Theorem~\ref{thm:cluster} then applies verbatim, and the zero-syndrome prune never removes a proper subset of $\supp(z)$. The value $w_K$ is itself found by Algorithm~\ref{alg:cluster} with the test $\mathbf 1_T\notin S$ dropped, since a nonzero minimum-weight element of $K$ has no proper zero-syndrome subset.
\end{proof}

The decomposition of Section~\ref{sec:decomp} interacts with Lemma~\ref{lem:conn} asymmetrically.

\begin{corollary}[colon connectivity]\label{cor:colonconn}
Let $z\in K$ satisfy $\pi(z)\ne0$ and $\wt(z)=d_{\mathcal C}$. Then $\sigma(T)\ne0$ for every nonempty $T\subsetneq\supp(z)$. Hence Algorithm~\ref{alg:cluster}, with the test $\mathbf 1_T\notin S$ replaced by $\pi(\mathbf 1_T)\ne0$, computes $d_{\mathcal C}$ exactly.
\end{corollary}
\begin{proof}
The map $(u,v)\mapsto v+(a)$ is linear on $K$. If $\sigma(T)=0$ then $\mathbf 1_T,z+\mathbf 1_T\in K$ are lighter than $z$ and $\pi(\mathbf 1_T)+\pi(z+\mathbf 1_T)=\pi(z)\ne0$, so one of them is a colon logical of weight below $d_{\mathcal C}$. A colon logical is automatically outside $S$, and the remaining steps are those of Theorem~\ref{thm:cluster}; translations preserve $\pi$ by Proposition~\ref{prop:equiv}.
\end{proof}

The annihilator analogue fails: for a minimum-weight $z$ with $\pi(z)=0$, a zero-syndrome split can produce two lighter pieces with $\pi(\mathbf 1_T)=\pi(z+\mathbf 1_T)\ne0$, both colon logicals, which contradicts nothing. We therefore obtain $d_{\mathcal A}$ either from the census at weight $d_Z$ or, when $d_{\mathcal A}>d_Z$, from Proposition~\ref{prop:irred} together with a one-sided witness.

\noindent\textbf{Cost.} Each call branches on at most $\wt(a)+\wt(b)-1$ qubits, so the search visits at most $2\sum_{j=1}^{W}(\wt(a)+\wt(b)-1)^{j-1}$ nodes, which is $O(5^{W})$ for weight-six checks and independent of $n$ apart from the syndrome bookkeeping. Each node updates $c$ syndrome bits and scans for the first unsatisfied check; the membership test in $S$ is run only at zero-syndrome nodes. The parity prune keeps the practical count far below the worst case (Table~\ref{tab:clustercert}). All arithmetic is exact over $\F_2$ and the integers, so the certificate involves no floating-point tolerance.

\begin{algorithm}[!t]
\caption{Search for low-weight logical witnesses (upper bounds)}\label{alg:screen}
\begin{algorithmic}[1]
\State \textbf{Input:} $\ell,m,a,b$; search radius $w_{\max}$. \textbf{Output:} low-weight witnesses with component and shape labels.
\State compute $\ann(a),\ann(b),(a\cln b),(b\cln a)$
\State search $t\in\ann(a)\setminus b\ann(a)$; record $(t,0)$ as an $\mathcal A$ witness
\State search $s\in\ann(b)\setminus a\ann(b)$; record $(0,s)$ and evaluate $\pi$ to label its component
\For{$v\in(a\cln b)\setminus(a)$ with $\wt(v)\le w_{\max}$}
  \State find a low-weight lift $u$ of $au=bv$; record $(u,v)$, $\pi[(u,v)]$, and its weight split
\EndFor
\State repeat the sparse lift search with $a$ and $b$ exchanged
\State \textbf{return} the lightest witnesses found; each gives an upper bound only
\end{algorithmic}
\end{algorithm}

\begin{algorithm}[!t]
\caption{Orbit-reduced support exclusion (lower bound)}\label{alg:supp}
\begin{algorithmic}[1]
\State \textbf{Input:} target $d_0$. \textbf{Output:} $d_Z\ge d_0$, or a logical support.
\For{each anchored support representative $E$ with $|E|<d_0$}
  \State $r\gets\operatorname{rank}\HX[:,E]$;\quad $\dim K(E)\gets|E|-r$
  \If{$\dim K(E)=0$} \textbf{continue} \EndIf
  \State compute $\dim S(E)$
  \If{$\dim K(E)\ne\dim S(E)$} \textbf{return} support $E$ \EndIf
\EndFor
\State \textbf{return} $d_Z\ge d_0$
\end{algorithmic}
\end{algorithm}

\subsection{Complexity of the structural and support steps}\label{sec:complexity}
Write $N=\ell m=n/2$. The vector spaces $\ann(a)$, $b\ann(a)$, $(a\cln b)$, $K$, and $S$ are obtained by Gaussian elimination on $O(N)\times O(N)$ matrices, hence in $O(N^3)$ arithmetic operations. This is only the \emph{structural} part. Extracting a low-weight element from a quotient or finding a minimum-weight lift remains a bounded or exact weight search; it is not covered by the $O(N^3)$ bound.

For support exclusion, anchoring one occupied coordinate at an origin reduces the number of weight-$i$ supports from $\binom ni$ to at most
\begin{equation}
\binom{n-1}{i-1}+\binom{N-1}{i-1},
\end{equation}
for the two block-anchor cases before any further deduplication. Since $\binom{n-1}{i-1}=(i/n)\binom ni$, the leading saving is of order $n/i$ rather than a uniform factor $N$. Rank pruning then removes supports whose restricted check columns are independent. The overall procedure remains exponential in the target distance; the benefit is symmetry and early rejection, not polynomial-time distance computation.

\subsection{Results of the lower-bound computations}\label{sec:milp}
Table~\ref{tab:clustercert} reports the cluster search at radius $W=d_{\mathrm U}-1$ for the six standard codes. Each run terminates without output, so Theorem~\ref{thm:cluster} gives $d_Z\ge d_{\mathrm U}$, and together with the witnesses of Table~\ref{tab:wit} and Proposition~\ref{prop:sym} the distance is determined in every case. The implementation was checked against exhaustive enumeration of $K\setminus S$ on $200$ random BB codes of length at most $32$, with no disagreement.

As an independent cross-check on the smaller codes we also solve the parity-constrained integer program
\begin{align}
\text{minimise}\quad & \mathbf 1^{\mathsf T} z \label{eq:milpobj}\\
\text{subject to}\quad & H_X z-2p=0,\quad Lz-y-2q=0,\quad \mathbf 1^{\mathsf T} y\ge1,\nonumber
\end{align}
with binary $z,y$ and nonnegative integer $p,q$, where $L\in\F_2^{k\times n}$ satisfies $\ker(L|_K)=S$ (for instance, rows given by a basis of $X$-logical operators). The factors of two enforce parity, so $y=Lz\bmod2$ and $\mathbf 1^{\mathsf T}y\ge1$ is exactly $z\notin S$. Solved with SciPy~1.17.0 and HiGHS~1.8.0 \cite{virtanen2020,huangfu2018}, with \texttt{mip\_rel\_gap}=0 and primal, dual, and MIP feasibility tolerances of $10^{-7}$, it returns the optima $4,6,10,10$ for the four shortest codes, in agreement with Table~\ref{tab:clustercert}. Because HiGHS works in floating point, we use this calculation only as a consistency check and base every reported lower bound on Algorithm~\ref{alg:cluster}.

\begin{table}[!t]
\centering\small
\caption{Anchored cluster search (Algorithm~\ref{alg:cluster}) at radius $W=d_{\mathrm U}-1$. ``Nodes'' counts calls of \textsc{Grow} over both anchors; no run produced a logical, so each proves $d\ge W+1$. Node counts are used as the implementation-independent measure of search effort.}
\label{tab:clustercert}
\renewcommand{\arraystretch}{1.18}
\setlength{\tabcolsep}{4pt}
\begin{tabular}{@{}lcrc@{}}
\toprule
code & $W$ & nodes & result \\
\midrule
$\code{18}{4}{4}$    & 3  & $19$                 & $d=4$ \\
$\code{72}{12}{6}$   & 5  & $171$                & $d=6$ \\
$\code{90}{8}{10}$   & 9  & $29{,}301$           & $d=10$ \\
$\code{108}{8}{10}$  & 9  & $24{,}863$           & $d=10$ \\
$\code{144}{12}{12}$ & 11 & $284{,}281$          & $d=12$ \\
$\code{288}{12}{18}$ & 17 & $455{,}504{,}708$    & $d=18$ \\
\bottomrule
\end{tabular}
\end{table}

\subsection{When the bounds meet: exact distance}
An explicit logical of weight $d_0$ together with a lower bound $d_Z\ge d_0$ gives $d_Z=d_0$, and then $d=d_Z$ by Proposition~\ref{prop:sym}. The witnesses of Table~\ref{tab:wit} and the certificates of Table~\ref{tab:clustercert} meet for all six standard codes. For $\code{72}{12}{6}$ and $\code{18}{4}{4}$ the support-exclusion test of Algorithm~\ref{alg:supp} gives the same lower bounds, but its cost grows like $\binom{n}{d_0-1}$ and it does not reach the longer codes.

\section{Examples and results}\label{sec:results}
The polynomial construction reproduces $n$ and $k$ for every standard code. Table~\ref{tab:wit} gives explicit nontrivial logical operators and Table~\ref{tab:clustercert} the matching lower bounds, so the distances $4,6,10,10,12,18$ are proved for all six examples. Figure~\ref{fig:bracket} and Table~\ref{tab:res} summarize these results. For $\code{288}{12}{18}$ this confirms a value reported as a numerical upper bound in \cite{bravyi2024}.

Running Algorithm~\ref{alg:cluster} at radius $d$ without stopping, and closing the output under translation, lists every minimum-weight logical operator. Grouping these operators by logical class and applying Definition~\ref{def:shape} to each class gives Table~\ref{tab:census}. Only $\code{108}{8}{10}$ attains its distance in a single component. In every other example both components attain $d$, and in $\code{90}{8}{10}$, $\code{144}{12}{12}$, and $\code{288}{12}{18}$ one-sided and balanced minimum-weight classes coexist. A colon class with a one-sided minimum occurs already in $\code{72}{12}{6}$, through representatives $(0,s)$ with $s\in\ann(b)\setminus(a)$.

\begin{table}[!t]
\centering\small
\caption{Census of minimum-weight logical classes, from the complete list of weight-$d$ logical operators. Each class is labelled by its component under $\pi$ and by its shape (Definition~\ref{def:shape}, $\tau=3$): one-sided (o), lopsided (l), balanced (b).}
\label{tab:census}
\renewcommand{\arraystretch}{1.18}
\setlength{\tabcolsep}{3.5pt}
\begin{tabular}{@{}lccccccc@{}}
\toprule
 & & \multicolumn{3}{c}{$\mathcal A$} & \multicolumn{3}{c}{$\mathcal C$} \\
\cmidrule(lr){3-5}\cmidrule(l){6-8}
code & classes & o & l & b & o & l & b \\
\midrule
$\code{18}{4}{4}$    & 15  & 0  & 3 & 0  & 0  & 12 & 0   \\
$\code{72}{12}{6}$   & 84  & 36 & 0 & 0  & 36 & 12 & 0   \\
$\code{90}{8}{10}$   & 72  & 9  & 0 & 0  & 9  & 18 & 36  \\
$\code{108}{8}{10}$  & 9   & 0  & 0 & 0  & 0  & 0  & 9   \\
$\code{144}{12}{12}$ & 246 & 36 & 0 & 12 & 0  & 72 & 126 \\
$\code{288}{12}{18}$ & 84  & 36 & 0 & 0  & 36 & 0  & 12  \\
\bottomrule
\end{tabular}
\end{table}

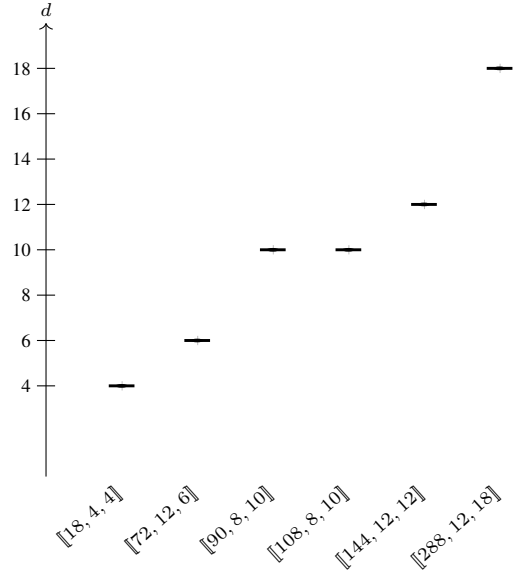
\begin{figure}[!t]
\centering
\begin{tikzpicture}[yscale=0.30,font=\scriptsize]
 \draw[->] (0,0)--(0,20) node[above=-1pt]{$d$};
 \foreach \d in {4,6,8,10,12,14,16,18}{\draw (-0.12,\d)--(0.12,\d); \node[left,inner sep=1pt] at (-0.16,\d){\d};}
 \foreach \x/\lab/\dL/\dU in {%
   1/{$\code{18}{4}{4}$}/4/4,
   2/{$\code{72}{12}{6}$}/6/6,
   3/{$\code{90}{8}{10}$}/10/10,
   4/{$\code{108}{8}{10}$}/10/10,
   5/{$\code{144}{12}{12}$}/12/12,
   6/{$\code{288}{12}{18}$}/18/18}{
   \draw[line width=3.5pt,gray!50] (\x,\dL)--(\x,\dU);
   \draw[line width=1.1pt] (\x-0.17,\dL)--(\x+0.17,\dL);
   \fill (\x,\dU) circle (2.3pt);
   \node[rotate=42,anchor=east,inner sep=1pt] at (\x,-0.5){\lab};
 }
\end{tikzpicture}
\caption{Distances of the six standard BB examples. In every case the explicit upper-bound witness meets the exact lower bound from the anchored cluster search.}
\label{fig:bracket}
\end{figure}

For the Gross code $\code{144}{12}{12}$ the annihilator screen returns the one-sided logical $u=(1+x^6)(1+y^2)a$ of weight $12$, and the cluster search at radius $11$ visits $284{,}281$ nodes without finding a lighter logical, so $d=12$. Its $246$ minimum-weight classes comprise one-sided and balanced annihilator classes together with lopsided and balanced colon classes; the one-sided witness is therefore representative of only a small part of the minimum-weight structure.

\begin{table*}[!t]
\centering\small
\caption{Results at a glance. $d_{\mathrm U}$ is the weight of an explicit logical, $d_{\mathrm L}$ the lower bound of Table~\ref{tab:clustercert}. Components attaining $d$ are taken from the census of Table~\ref{tab:census}. Published values are for comparison only; for the length-$288$ code Ref.~\cite{bravyi2024} reports $18$ as a numerical upper bound.}
\label{tab:res}
\renewcommand{\arraystretch}{1.2}
\begin{tabular}{@{}lccccccl@{}}
\toprule
code & $n$ & $k$ & components attaining $d$ & $d_{\mathrm U}$ & $d_{\mathrm L}$ & published & status \\
\midrule
$\code{18}{4}{4}$    & 18  & 4  & $\mathcal A,\mathcal C$ & 4  & 4  & 4       & proved \\
$\code{72}{12}{6}$   & 72  & 12 & $\mathcal A,\mathcal C$ & 6  & 6  & 6       & proved \\
$\code{90}{8}{10}$   & 90  & 8  & $\mathcal A,\mathcal C$ & 10 & 10 & 10      & proved \\
$\code{108}{8}{10}$  & 108 & 8  & $\mathcal C$ only       & 10 & 10 & 10      & proved \\
$\code{144}{12}{12}$ & 144 & 12 & $\mathcal A,\mathcal C$ & 12 & 12 & 12      & proved \\
$\code{288}{12}{18}$ & 288 & 12 & $\mathcal A,\mathcal C$ & 18 & 18 & $\le18$ & proved \\
\bottomrule
\end{tabular}
\end{table*}

\subsection{Decomposition profile}\label{sec:profile}
The structural dimensions require only Gaussian elimination; $w_{\ann}$ is an exhaustive minimum over $\ann(a)\setminus b\ann(a)$, and $d_{\mathcal A},d_{\mathcal C}$ are the exact component minima of Theorem~\ref{thm:ddec}, obtained from the census and, for $\code{108}{8}{10}$, from Proposition~\ref{prop:irred} with $w_K=6$ and radius $11$. The length-$288$ value $w_{\ann}=18$ was independently checked by exhaustive enumeration of all $2^{24}$ elements of $\ann(a)$.

\begin{table*}[!t]
\centering\small
\caption{Structural dimensions and exact component minima. Lemma~\ref{lem:frob} guarantees $r_{\mathcal A}=r_{\mathcal C}=k/2$ in every row, and $d=\min(d_{\mathcal A},d_{\mathcal C})$ by Theorem~\ref{thm:ddec}.}
\label{tab:profile}
\renewcommand{\arraystretch}{1.18}
\begin{tabular}{@{}lcccccccc@{}}
\toprule
code & $\dim\ann(a)$ & $\dim b\ann(a)$ & $r_{\mathcal A}$ & $r_{\mathcal C}$ & $w_{\ann}$ & $d_{\mathcal A}$ & $d_{\mathcal C}$ & $d$ \\
\midrule
$\code{18}{4}{4}$    & 2  & 0  & 2 & 2 & 6  & 4  & 4  & 4  \\
$\code{72}{12}{6}$   & 12 & 6  & 6 & 6 & 6  & 6  & 6  & 6  \\
$\code{90}{8}{10}$   & 6  & 2  & 4 & 4 & 10 & 10 & 10 & 10 \\
$\code{108}{8}{10}$  & 6  & 2  & 4 & 4 & 12 & 12 & 10 & 10 \\
$\code{144}{12}{12}$ & 12 & 6  & 6 & 6 & 12 & 12 & 12 & 12 \\
$\code{288}{12}{18}$ & 24 & 18 & 6 & 6 & 18 & 18 & 18 & 18 \\
\bottomrule
\end{tabular}
\end{table*}

\begin{remark}[naive bounds are weak or false]\label{rem:naive}
Two natural bounds are worth recording only as cautions, since the decomposition explains both failures. Mapping a syzygy to $w=au=bv\in(a)\cap(b)$ and bounding $\wt(u)+\wt(v)$ from $\wt(w)$ is valid but lossy: it is blind to the one-sided family, where $w=0$, and the convolution step loses a constant factor, returning $3$ against $d=6$ on $\code{72}{12}{6}$. A difference-set argument, which would assert $\wt(u)\wt(v)\ge|\supp(a)-\supp(b)|$, is in fact false over $\R$: for $\code{18}{4}{4}$ with $b=a^2$ it predicts weight at least $5$, whereas $(a,1)$ is a genuine logical of weight $4$. Both fail on support shapes that the quotient decomposition helps expose, but those shapes are not themselves the two components of Theorem~\ref{thm:exact}.
\end{remark}

\subsection{Screening on random pairs}\label{sec:sweep}
To show the screens behave as a filter rather than a post-hoc explanation, we sampled random weight-three pairs $(a,b)$ on two tori. The setup is as follows: each of $a,b$ is an independent uniform weight-three polynomial (distinct monomials within a polynomial, but $a$ and $b$ may share monomials); pairs are not deduplicated under translation or block swap; disconnected cases are not excluded; the seeds are fixed ($3$ for $(6,6)$, $5$ for $(9,6)$) for reproducibility. We keep the pairs with $k>0$ and report, for the one-sided annihilator screen, the fraction with a witness of weight at most $6$ and at most $12$, and the fraction exhibiting a sparse colon shortcut $v\in(a\cln b)\setminus(a)$ with $\wt(v)\le2$ (a lopsided relation). Results are in Table~\ref{tab:sweep}. The dimension is concentrated at small values: at $(6,6)$ the $k>0$ pairs split as $k=4$ ($28$), $k=8$ ($13$), $k=24$ ($1$); at $(9,6)$ as $k=4$ ($36$), $k=8$ ($3$). The two screens trade off with the torus: at $(6,6)$ the one-sided screen catches many pairs, while at $(9,6)$ it rarely fires below weight $6$ but the lopsided-shortcut screen catches most. The marginal rates show that each diagnostic fires often on these unconstrained samples. Because Table~\ref{tab:sweep} does not report the intersections of the diagnostics, it does not establish the fraction of pairs surviving all screens and we make no rarity claim for the curated standard codes. We present the counts only as a fixed-seed demonstration of screening behaviour; the entries are exact for the stated seeds. These random pairs are not intended to model optimized BB-code searches; they only show that the screens identify many early failure modes in an unconstrained candidate pool.

\begin{table}[!t]
\centering\small
\caption{Screening on $200$ random weight-three pairs per torus. ``$\le6$''/``$\le12$'' give the fraction of $k>0$ pairs whose one-sided annihilator screen returned a witness of that weight; ``lop'' is the fraction with a sparse colon shortcut.}
\label{tab:sweep}
\renewcommand{\arraystretch}{1.2}
\setlength{\tabcolsep}{4pt}
\begin{tabular}{@{}lccccc@{}}
\toprule
torus $(\ell,m)$ & samples & $k{>}0$ & $\le6$ & $\le12$ & lop \\
\midrule
$(6,6)$ & 200 & 42 & $45\%$ & $71\%$ & $57\%$ \\
$(9,6)$ & 200 & 39 & $3\%$  & $28\%$ & $74\%$ \\
\bottomrule
\end{tabular}
\end{table}

\subsection{Search and design procedure}\label{sec:workflow}
Put together, the structural decomposition and the support-shape diagnostics give a concrete recipe. Work proceeds from cheap tests to expensive ones, and the output is a pair of bounds with an explicit label rather than a single unqualified number. Given $(\ell,m,a,b)$:
\begin{enumerate}[leftmargin=1.5em,itemsep=1pt]
\item \emph{Set up.} Build $A,B$ over $\F_2$ and compute $k$, together with $\ker A$, $B(\ker A)$, $\operatorname{im}A$, and $B^{-1}(\operatorname{im}A)$; these give $r_{\mathcal A}$ and $r_{\mathcal C}$ (Table~\ref{tab:profile}).
\item \emph{Reject sparse shortcuts.} Test for sparse $r,s$ with $ar=bs$; reject if $\wt(r)+\wt(s)<d_{\text{target}}$. This catches lopsided failures such as $b=a^2$.
\item \emph{Screen annihilators.} Compute the spaces $\ann(a)$, $b\,\ann(a)$ by Gaussian elimination and search for a low-weight representative $w_{\ann}^{\mathrm{found}}$; reject candidates whose $w_{\ann}^{\mathrm{found}}$ is already below the design target.
\item \emph{Sparse relation screen.} Search sparse representatives in $(a\cln b)/(a)$ and $(b\cln a)/(b)$ and lift them; record the component label and weight split of every witness found.
\item \emph{Prove the lower bound.} Run the anchored cluster search (Algorithm~\ref{alg:cluster}) at radius $d_{\text{target}}-1$ on the competitive candidates; for $d_{\mathcal C}$ alone use the colon version of Corollary~\ref{cor:colonconn}.
\item \emph{Report the result.} Record $d_{\mathrm L}\le d\le d_{\mathrm U}$ together with the component and shape data. When $d_{\mathrm L}=d_{\mathrm U}$ the distance is established; otherwise continue the exact search or report the remaining interval.
\end{enumerate}

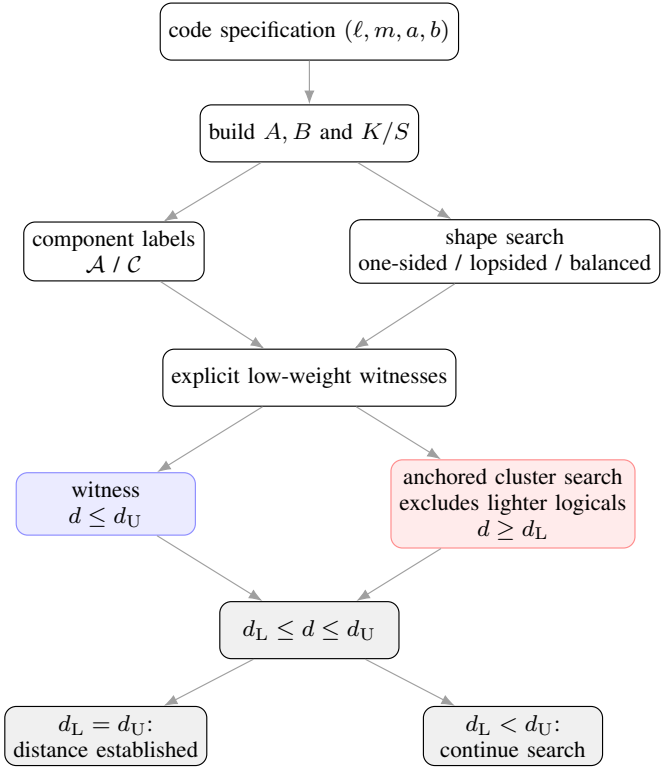
\begin{figure}[!t]
\centering
\resizebox{0.98\columnwidth}{!}{%
\begin{tikzpicture}[font=\footnotesize,
  b/.style={draw,rounded corners,align=center,inner sep=3pt,minimum height=7mm,minimum width=22mm},
  up/.style={b,fill=blue!8,draw=blue!45},
  lo/.style={b,fill=red!8,draw=red!45},
  res/.style={b,fill=black!6},
  ar/.style={-{Latex},gray!75}]
 \node[b] (spec) at (0,0) {code specification $(\ell,m,a,b)$};
 \node[b] (mats) at (0,-1.25) {build $A,B$ and $K/S$};
 \node[b] (comp) at (-2.4,-2.7) {component labels\\$\mathcal A$ / $\mathcal C$};
 \node[b] (shape) at (2.4,-2.7) {shape search\\one-sided / lopsided / balanced};
 \node[b] (cand) at (0,-4.25) {explicit low-weight witnesses};
 \node[up] (upb) at (-2.5,-5.8) {witness\\$d\le d_{\mathrm U}$};
 \node[lo] (lob) at (2.5,-5.8) {anchored cluster search\\excludes lighter logicals\\$d\ge d_{\mathrm L}$};
 \node[res] (out) at (0,-7.35) {$d_{\mathrm L}\le d\le d_{\mathrm U}$};
 \node[res] (pv) at (-2.5,-8.65) {$d_{\mathrm L}=d_{\mathrm U}$:\\ distance established};
 \node[res] (bk) at (2.5,-8.65) {$d_{\mathrm L}<d_{\mathrm U}$:\\ continue search};
 \draw[ar] (spec)--(mats);
 \draw[ar] (mats)--(comp);
 \draw[ar] (mats)--(shape);
 \draw[ar] (comp)--(cand);
 \draw[ar] (shape)--(cand);
 \draw[ar] (cand)--(upb);
 \draw[ar] (cand)--(lob);
 \draw[ar] (upb)--(out);
 \draw[ar] (lob)--(out);
 \draw[ar] (out)--(pv);
 \draw[ar] (out)--(bk);
\end{tikzpicture}}
\caption{Workflow from code specification to distance bounds. Algebraic component and representative shape are evaluated separately; explicit witnesses provide upper bounds, while the anchored cluster search provides the lower bound.}
\label{fig:pipeline}
\end{figure}

Fig.~\ref{fig:pipeline} shows the whole route from a code specification to its bounds.
The same classification also yields design rules: avoid $b=a^{2^t}$ and other sparse Frobenius relations; avoid $a$ or $b$ with low-weight annihilators; penalize candidate pairs with sparse colon representatives; and prefer pairs whose cluster-search lower bound $d_{\mathrm L}$ is close to the best witness $d_{\mathrm U}$.

\section{Discussion and limitations}\label{sec:disc}
Theorem~\ref{thm:exact} splits logical classes into two projection-defined components, and Theorem~\ref{thm:ddec} turns the split into a componentwise distance formula. Neither yields a closed-form number, and the census shows why none should be expected from component labels alone: the minimum is attained in both components for five of the six examples, with one-sided, lopsided, and balanced classes coexisting. What the decomposition does provide is a second invariant alongside weight, and Corollary~\ref{cor:colonconn} shows that this invariant is compatible with cluster verification on the colon side.

\subsection{Results and open problems}\label{sec:disc-open}
All six standard distances are proved here, and the census determines $d_{\mathcal A}$ and $d_{\mathcal C}$ separately. The worst-case search grows exponentially with the target distance. For the examples studied here, the parity and translation reductions keep the observed node counts far below the naive $5^d$ growth, including $455{,}504{,}708$ nodes for the length-$288$ code at radius $17$. Extending the method to substantially larger distances will likely require additional symmetry reduction, for example canonicalisation under a larger automorphism group. Two questions remain open. The first is an annihilator analogue of Corollary~\ref{cor:colonconn}, which would give $d_{\mathcal A}$ directly without passing through the full census. The second is whether the counts in Table~\ref{tab:census}, for instance the $36$ one-sided annihilator classes appearing in three of the codes, can be predicted from the structure of $\ann(a)$ and $(a\cln b)$; classical defining-set and apparent-distance tools \cite{bernal2016,macwilliams1977} are natural candidates for the one-sided part.

\subsection{Scope and implementation}\label{sec:scope}
The quantity computed here is the minimum Hamming weight of a nontrivial logical operator. That is a property of the code, not a prediction of performance: it should not be read as a logical error rate under a particular decoder, a biased or spatially correlated noise model, measurement error, or full circuit-level noise. Those depend on the decoder and the noise, and on many operators rather than the lightest one.

With that scope fixed, three consequences are worth stating. A weight-$d$ logical operator is an explicit minimum-weight undetectable error pattern, so the operators in Appendix~\ref{app:wit} name the specific qubit sets a code is least able to protect. The same pair $a=x^3+y+y^2$, $b=y^3+x+x^2$ defines $\code{72}{12}{6}$ on the $6\times6$ torus and $\code{108}{8}{10}$ on the $9\times6$ torus, yet the first has one-sided minima in both components while every minimum of the second is a balanced colon class. The minimum-weight structure is therefore a property of the pair together with the torus, and the census, not the polynomials alone, is what exposes it. The structural spaces used by the screens are obtained in $O(n^3)$ time, after which bounded low-weight searches can reject many weak candidates before a decoder simulation is run; the bounded searches themselves have the additional combinatorial cost stated in Section~\ref{sec:complexity}. We do not claim the classification predicts logical error rates; establishing that would require the simulations this paper does not do.

\subsection{Beyond distance: further applications}\label{sec:disc-apps}
The decomposition is not specifically related to the distance problem and code screening, and we expect it to be useful beyond the two applications developed here. Because it is stated for any relation $au+bv=0$ over a group algebra, the same exact sequence applies to any two-block group algebra code, not only the BB family, and to other CSS constructions whose checks are defined by a pair of ring elements. The decomposition also gives explicit logical subspaces and quotient coordinates beyond a single distance value. Combined with a separate support-shape analysis, these data may help identify logical representatives suited to hardware-efficient operations, complementing explicit gate constructions such as those of Eberhardt and Steffan \cite{eberhardt2024}. The same origin of the lightest operators could also inform decoder design, for instance by weighting a decoder toward the structurally cheaper failure mode identified by the profile of Section~\ref{sec:profile}, or guide large-scale code search, where the polynomial-time structural stage can be followed by bounded sparse searches before an exact minimum-weight calculation is attempted. None of these directions is developed in this paper; we flag them as natural next steps.

\appendices
\section{Software and Data}\label{sec:soft}
Code and data supporting the results in this paper are available at \url{https://github.com/mohammad-rowshan/BB-Codes-Distance-Logical-Decomposition}. 

\section{Code data and explicit witnesses}\label{app:wit}
Table~\ref{tab:codedata} gives the defining polynomials and a verified distance witness for each standard code. Every pair satisfies $au+bv=0$ over $\F_2$ and is a non-stabilizer element of $K\setminus S$, confirmed by exact $\F_2$ computation. The table lists convenient representatives; component membership is determined separately by $\pi[(u,v)]=v+(a)$.

\begin{table}[!t]
\centering\scriptsize
\caption{Code data and convenient minimum-weight witnesses.}
\label{tab:codedata}
\renewcommand{\arraystretch}{1.2}
\setlength{\tabcolsep}{3pt}
\begin{tabular}{@{}>{\raggedright\arraybackslash}p{1.50cm}>{\raggedright\arraybackslash}p{2.55cm}>{\raggedright\arraybackslash}p{3.10cm}@{}}
\toprule
code $(\ell,m)$ & $a,\ b$ & witness $(u,v)$ \\
\midrule
$\code{18}{4}{4}$ $(3,3)$ & $a=x+y+xy$, $b=a^2$ & $u=a,\ v=1$ \\
$\code{72}{12}{6}$ $(6,6)$ & $a=x^3+y+y^2$, $b=y^3+x+x^2$ & $u=(1+y^2)a,\ v=0$ \\
$\code{90}{8}{10}$ $(15,3)$ & $a=x^9+y+y^2$, $b=1+x^2+x^7$ & $u=(1{+}y)\sum_{j=0}^{4}x^{3j},\ v=0$ \\
$\code{108}{8}{10}$ $(9,6)$ & $a=x^3+y+y^2$, $b=y^3+x+x^2$ & see \eqref{eq:w108} \\
\shortstack[l]{$\code{144}{12}{12}$\\$(12,6)$} & $a=x^3+y+y^2$, $b=y^3+x+x^2$ & $u=(1{+}x^6)(1{+}y^2)a,\ v=0$ \\
\shortstack[l]{$\code{288}{12}{18}$\\$(12,12)$} & $a=x^3+y^2+y^7$, $b=y^3+x+x^2$ & see~\eqref{eq:w288}, $v=0$ \\
\bottomrule
\end{tabular}
\end{table}

The displayed minimum witness for $\code{108}{8}{10}$ is the balanced case among the six examples, so we give it in full. On $\ell=9,m=6$ with $a=x^3+y+y^2$, $b=y^3+x+x^2$,
\begin{align}
u&=x+xy+x^4y^4+x^5+x^7y^3+x^8y^4,\nonumber\\
v&=xy^4+x^3y+x^4y^2+x^6y^5.\label{eq:w108}
\end{align}
Its class is one of the nine minimum-weight classes of the code, all of which are $3$-balanced colon classes (Example~\ref{ex:108} and Table~\ref{tab:census}).

The $\code{288}{12}{18}$ one-sided witness (weight $18$, $v=0$) is
\begin{align}
u\;=&\;y+y^3+y^4+y^6+y^8+y^9+y^{10}+y^{11}\nonumber\\
&+x^3y^2+x^3y^4+x^3y^6+x^3y^8\nonumber\\
&+x^6y^2+x^6y^6+x^6y^7+x^6y^{11}+x^9+x^9y^4,\label{eq:w288}
\end{align}
which is verified by $Au=0$ over $\F_2$, $u\notin b\,\ann(a)$ (so the class is a genuine logical, not a stabiliser), and $\wt(u)=18$.

\section*{Acknowledgment}
This research work was conducted in part at the University of New South Wales (UNSW) Sydney.

\bibliographystyle{IEEEtran}
\bibliography{refs}
\end{document}